\documentclass[a4paper,12pt,english]{article}
\usepackage[utf8]{inputenc}   
\usepackage{lmodern}
\usepackage[english]{babel}   
\usepackage[babel]{microtype}
\usepackage{amsfonts,amsmath,latexsym,amsthm,fixmath}
\usepackage{amssymb,ifthen}
\usepackage{authblk}
\usepackage{enumitem}
\usepackage{color,graphicx}
\usepackage{url}
\usepackage[
    pdftex,
    pdfstartview={XYZ 0 1000 1.0},
    bookmarks=false,
    colorlinks,
    breaklinks=true,
    citecolor=citecolor,
    filecolor=blue,
    linkcolor=linkcolor,
    urlcolor=urlcolor,
    pdfborder={0 0 0}
    ]{hyperref}
\definecolor{urlcolor}{rgb}{0, 0.5, 0}
\definecolor{citecolor}{rgb}{.5,0,.25}
\definecolor{linkcolor}{rgb}{0,0,1}
\usepackage{geometry}
\usepackage{algorithm,algorithmic}

\newtheorem{theorem}{Theorem}
\newtheorem{lemma}[theorem]{Lemma}
\newtheorem{claim}[theorem]{Claim}
\newtheorem{corollary}[theorem]{Corollary}
\newtheorem{proposition}[theorem]{Proposition}
\newtheorem{remark}[theorem]{Remark}

\newcommand{\includesvg}[2][]{%
\def\tempa{#1}\def\tempb{}%
\ifx\tempa\tempb\else\let\svgwidth\tempa\fi
\input{\svgpath#2.pdf_tex}%
}

\graphicspath{{./figures/}}
\def\svgpath{figures/}

\newcommand{\mC}{\mathbb{C}}
\newcommand{\mZ}{\mathbb{Z}}

\newcommand{\mR}{\mathbb{R}}

\newcommand{\mQ}{\mathbb{Q}}

\def\dist{\operatorname{dist}}

\def\conv{\operatorname{conv}}
\def\Re{\operatorname{Re}}

\newcommand{\Norm}[1]{N_{#1}} 
\newcommand{\Ball}[1]{B_{#1}} 
\newcommand{\SimpleC}{\mathcal{SC}_{G}}
\newcommand{\TSimpleC}{\mathcal{TSC}_{G,w}}
\newcommand{\HoneT}{H_1(T;\mZ)}
\newcommand{\ie}{i.e.}
\newcommand{\define}[1]{\emph{#1}}

\title{Algorithms for Length Spectra of Combinatorial Tori} 

\author[1]{Vincent Delecroix}
\author[2]{Matthijs Ebbens}
\author[3]{\\Francis Lazarus \thanks{This author is partially supported by the French ANR projects MINMAX (ANR-19-CE40-0014) and the LabEx PERSYVAL-Lab (ANR-11-LABX-0025-01) funded by the French program Investissement d’avenir.}}
\author[1]{Ivan Yakovlev}

\affil[1]{Univ. Bordeaux, CNRS, Bordeaux INP, LaBRI, UMR 5800, F-33400 Talence, France}
\affil[2]{Institut Fourier, CNRS, Universit\'e Grenoble Alpes, France}
\affil[3]{G-SCOP/Institut Fourier, CNRS, Universit\'e Grenoble Alpes, France}

\begin{document}

\maketitle

\begin{abstract}
Consider a weighted, undirected graph cellularly embedded on a topological
surface. The function assigning to each free homotopy class of closed
curves the length of a shortest cycle within this homotopy class is called
the marked length spectrum. The (unmarked) length spectrum is obtained by
just listing the length values of the marked length spectrum in increasing
order.

In this paper, we describe algorithms for computing the (un)marked length
spectra of graphs embedded on the torus. More specifically, we preprocess a
weighted graph of complexity $n$ in time $O(n^2 \log \log n)$ so that,
given a cycle with $\ell$ edges representing a free homotopy class, the
length of a shortest homotopic cycle can be computed in $O(\ell+\log n)$
time. Moreover, given any positive integer $k$, the first $k$ values of its
unmarked length spectrum can be computed in time $O(k \log n)$.

Our algorithms are based on a correspondence between weighted graphs on
the torus and polyhedral norms. In particular, we give a weight independent
bound on the complexity of the unit ball of such norms. As an immediate
consequence we can decide if two embedded weighted graphs have the same
marked spectrum in polynomial time. We also consider the problem of comparing the unmarked spectra and provide a polynomial time algorithm in the unweighted case and a randomized polynomial time algorithm otherwise.
\end{abstract}

\section{Introduction}\label{sec:introduction}
Combinatorial surfaces are well-studied in computational topology and are usually represented as graphs cellularly embedded on a topological surface. Given a combinatorial surface $S$ with underlying graph $G$, many algorithms were proposed for computing the length of its shortest homotopically non-trivial closed walk~\cite{thomassen1990,erickson2004,k-csntc-06,cabello2007,cabello2012,e-cocb-12,cce-msspe-13}. Here, the length of a walk is the sum of the weights of its edges if the edges are weighted, or the number of edges if not. However, relatively little is known about how to compute the \emph{second shortest} non-trivial closed walk, the \emph{third shortest} and so on. More precisely, for every closed walk $c$ in $G$, we can compute the length of the shortest closed walk freely homotopic to $c$ on $S$. By definition, this length only depends on the free homotopy class of $c$. The ordered sequence of lengths of all free homotopy classes of closed walks is called the \emph{length spectrum} of $S$ with respect to its (weighted) graph $G$, while the mapping between free homotopy classes of curves and their lengths is called the \emph{marked length spectrum}. The marked length spectrum thus records for every length in the sequence from which free homotopy class it comes from. These notions are well studied in the
realm of hyperbolic or Riemannian surfaces~\cite{o-smlsc-90,b-gscrs-92,p-islsq-18}.
A striking result in that respect is that the marked length spectrum of a non-positively curved surface entirely determines the geometry of the surface~\cite{o-smlsc-90}.  In
 other words, one may learn the geometry of a surface by just looking at the length of its curves. However, the unmarked length spectrum does not determine the surface even in constant curvature~\cite{vigneras}.

 Analogously, Schrijver~\cite[Th. 1]{schrijver1992} proved that embedded graphs that are minor-minimal among graphs with the same marked length spectrum, which he calls \emph{kernels}, are determined by their marked length spectrum up to simple transformations, namely taking the dual graph and applying a sequence of so-called $\Delta Y$ transformations.

In a subsequent paper~\cite{schrijver1993} Schrijver restricts to unweighted graphs embedded on the torus and notices that the marked length spectrum extends to an integer norm in $\mR^2$, i.e., a norm taking integer values at integer vectors.
Moreover, its dual unit ball is a finite polygon whose vertices have integer coordinates (see also~\cite{s-ntivi-16,s-inofc-20}). This allows him to reconstruct for every integer norm a graph whose marked length spectrum is given by this norm.

The aim of this paper is threefold. We first extend the results of Schrijver~\cite{schrijver1993} to weighted graphs on the torus.
There are good reasons to focus on the torus. For instance, the marked length spectrum of a graph embedded on the torus being a norm is due to the equivalence between homotopy and homology, which is not true for higher genus surfaces. For weighted graphs, the marked length spectrum function still extends to a norm on $\mR^2$, that we denote by $\Norm{G,w}$, but not necessarily to an integer norm. However, we show that it is a polyhedral norm for any choice of weights. 
In other words, the unit ball $\Ball{G,w} := \{\alpha \in \mR^2\mid \Norm{G,w}(\alpha) \le 1\}$ is always a polygon. We also prove that the number of extremal points of this polygon  is bounded by a linear function of the number of vertices of $G$ that is independent of the weights $w$. 
\begin{theorem}
  \label{thm:weightedgraphspolyhedralnorm}
  Let $(G,w)$ be a weighted graph with $|V|$ vertices cellularly embedded on the torus. Then $\Norm{G,w}$ is a polyhedral norm. Moreover, its unit ball $\Ball{G,w}$ is a polygon with no more than $4|V|+5$ extremal points, and the ratio of the coordinates of each extremal point is rational.
\end{theorem}
We also extend Schrijver's reconstruction of a toroidal graph from an integer norm~\cite{schrijver1993} to the weighted case for non-integer polyhedral norms. See Theorem~\ref{thm:normReconstruction} in Section~\ref{sec:polyhedralnorms}.

Our second goal is to provide algorithms to compute the unit ball $\Ball{G,w}$ and to compute the length spectrum. All our complexity estimates assume the standard RAM model of computation or the standard real-RAM model for non-integer weights supporting constant time arithmetic operations. We denote by $n$ the complexity of $G$, that is its total number of edges and vertices.
\begin{theorem}\label{th:unit-ball}
  The unit ball $\Ball{G,w}$ can be computed in $O(n^2\log\log n)$ time, or in $O(n^2)$ time if $G$ is unweighted.
\end{theorem}
After this preprocessing step, we can compute the length of a shortest closed walk freely homotopic to an input closed walk of $\ell$ edges in $O(\ell+\log n)$ time. It is a priori not obvious to sort the values of the length spectrum from their homotopy classes. However, by decomposing the unit ball into unimodular sectors, i.e., sectors generated by the columns of unimodular matrices, we are able to compute efficiently
the first $k$ values of the length spectrum.
\begin{theorem} \label{thm:lengthSpectrumComputation}
  Let $(G,w)$ be a weighted graph of complexity $n$ cellularly embedded on the torus and let $k$ be a positive integer. After $O(n^2\log\log n)$ preprocessing time, the first $k$ values of the length spectrum of $(G,w)$ can be computed in $O(k\log n)$ time. If $G$ is unweighted, the $\log\log n$ factor in the  preprocessing time can be omitted. 
\end{theorem}

Recently, two of the co-authors~\cite{ebbens2022} used shortest path computations in the universal cover of the torus to determine the length spectrum. In this way, they were able to compute the first $k$ values of the length spectrum in time $O(kn^2\log(kn))$. This is to be compared to $O(k\log n)$ in the present paper.

 Finally, as a consequence of Theorem~\ref{th:unit-ball}, we provide algorithms to check whether two weighted graphs have the same marked or unmarked length spectrum. In the unweighted case it takes the following simple form.
 \begin{theorem}
 \label{thm:decidingEqualitySpectrumUnweighted}
 The equality of marked and unmarked spectra of two unweighted graphs $G$ and $G'$ embedded on tori can be tested in time $O(n^2)$ and $O\left(n^3\right)$, respectively.
\end{theorem}
 Our algorithm for the marked length spectrum is also polynomial in the weighted case (Theorem~\ref{thm:decidingEqualityMarkedSpectrum}). However, to compute the unmarked length spectrum we reduce the equality of length spectra to polynomial identity testing (PIT), which is known to lie in the co-RP complexity class of problems whose negation admit a randomized polynomial-time algorithm~\cite{schwartz}. See Theorem~\ref{thm:decidingEqualityUnmarkedSpectrum}.
 It becomes deterministic polynomial in the unweighted case as stated in Theorem~\ref{thm:decidingEqualitySpectrumUnweighted} above.
 
 While the equality of marked length spectra implies some kind of structural equivalence between kernels~\cite{schrijver1992}, we provide an  example of isospectral toroidal graphs whose associated unit balls are not related by any linear transformation. In particular, they cannot have the same marked length spectrum even after applying a homeomorphism on the torus.

\paragraph*{Organization of the paper}
We start by discussing some preliminaries in Section~\ref{sec:preliminaries}. We next prove Theorem~\ref{thm:weightedgraphspolyhedralnorm} in Section~\ref{sec:polyhedralnorms}. Theorem~\ref{th:unit-ball} is the object of Sections~\ref{sec:short-basis} and~\ref{sec:computingUnitBall}, while Theorem~\ref{thm:lengthSpectrumComputation} is proved in Section~\ref{sec:lengthSpectrumComputation}. The equality of length spectra is finally discussed in Section~\ref{sec:lengthSpectraEquality}.

\section{Preliminaries}\label{sec:preliminaries}
Let $G=(V,E)$ be an undirected graph with vertex set $V$ and edge set $E$. We allow $G$ to have loop edges and multiple edges.  We denote by $n:= |V|+|E|$ the complexity of $G$. A \define{weight function} for $G$ is a map $w : E\to \mR_+$. The positive value $w(e)$ is  the \define{weight} (or \define{length}) of the edge $e\in E$. We write $(G,w)$ for a graph $G$ with  weight function $w$. A \define{walk} is a finite alternating sequence of vertices and oriented edges, starting and ending with a vertex, such that each vertex is the target (source) of the preceding (succeeding) edge in the sequence. We also use \define{path} as a synonym for walk. The \define{length} $w(c)$ of a walk $c$ is the sum of the weights of its edges, counted with multiplicity. A walk is \define{closed} when its first and last vertices coincide. This vertex is the \define{basepoint} of the closed walk. A closed walk without repeated vertices is also called a \define{simple cycle}. 

Throughout this paper, we will use $S$ to denote a topological surface and $T$ to denote the topological orientable surface of genus 1, i.e., a torus. In this paper we assume that $G$ is \define{cellularly embedded} on $S$, which means that the complement $S\setminus G$ is a collection of open disks. This embedding can be represented using one of the standard representations, e.g., the incidence graph of flags~\cite{eppstein2003} or rotation systems~\cite{mohar2001}. A surface together with a cellular embedding of a weighted graph is called a \define{combinatorial surface}.

We now recall some definitions from the topology of surfaces.
We refer the reader to Stillwell~\cite{stillwell1993} for basic notions in algebraic topology.
\paragraph*{Homotopy}
Two walks of $G$ are said \define{homotopic} if they are homotopic as curves in $S$, i.e., one can be continuously deformed into the other on $S$ while keeping the endpoints fixed. Similarly, two closed walks are \define{freely homotopic} if they are so as curves in $S$. Here, we do not require the basepoint to stay fixed during the homotopy. Closed walks (freely) homotopic to a walk reduced to a vertex are said \define{trivial}.  Homotopy defines an equivalence relation between walks. The set of homotopy classes of closed walks with fixed basepoint $v$ defines a group under concatenation. It is called the \define{fundamental group} of $S$ and denoted by $\pi_1(S,v)$.

The fundamental group of the torus is Abelian and isomorphic to $\mZ^2$. See e.g.~\cite{stillwell1993}. $\pi_1(T,v)$ is thus in bijection with its set of conjugacy classes, hence with the set of free homotopy classes.

A closed walk is  \define{tight} if it is shortest in its free homotopy class. Note that a homotopy class may contain more than one tight closed walk. Let $\cal C$ denotes the set of free homotopy classes of $S$. The map $\mathcal{C}\to \mR_+$ that associates to every free homotopy class the length of a tight closed walk in the class is the \define{marked length spectrum} of $S$ with respect to $(G,w)$. The \define{unmarked length spectrum} is  the list containing in increasing order the lengths of the non-trivial free homotopy classes of $G$, counted  with multiplicity: if two homotopy classes have the same length, then this length will appear twice in the list. Remark that the first element of the length spectrum is commonly referred to as the \define{systole}.

\paragraph*{Homology} The fundamental group is not the only interesting group that one may associate to a surface. Homology groups have a slightly more abstract definition but are actually simpler to compute and to deal with. Let $F$ be the set of faces of the cellular embedding of $G$ in $S$. We also call a face, an edge or a vertex, a $k$-cell for $k=2,1,0$, respectively. The group of \define{2-chains}, $C_2$, is the group of formal linear combinations of faces with integer coefficients with the obvious addition as group operation. A typical element of $C_2$ has the form $\Sigma_{f\in F}n_ff$ with $n_f\in \mZ$.  Likewise, the group $C_1$ of 1-chains  and the group $C_0$ of 0-chains  are the groups of formal linear combinations of edges and vertices, respectively. Cells are assumed to be oriented, and a cell multiplied by $-1$ represents the same cell with opposite orientation.

For $k=1,2$, the boundary operator $\partial_k: C_k\to C_{k-1}$ is the linear extension of the map that sends a $k$-cell to the formal sum of its boundary facets, where the coefficient of a facet in the sum
is 1 if its orientation is induced by the orientation of the $k$-cell and $-1$ otherwise. 
The kernel of $\partial_k$ is a subgroup of $C_k$ denoted by $Z_k$. Its elements are called \define{$k$-cycles}, not to be confused with cycles in the graph theoretical sense. The image of $\partial_k$ is the subgroup  of   \define{$k$-boundaries} of $C_{k-1}$ and denoted by $B_{k-1}$. 
The \define{first homology group} of $S$ with respect to the coefficients $\mZ$ is 
the group $H_1(S;\mZ):= \ker\partial_1/\operatorname{Im} \partial_2$ of $1$-cycles modulo the group of $2$-boundaries. From homology theory, we know that $H_1(S;\mZ)$ does not depend on the specific cell decomposition induced by the cellular embedding of $G$. We can similarly define the first homology group with real coefficients $H_1(S;\mR)$. Since the 1-chains only depend on the graph $G$, we also write $Z_1(G;\mZ)$ for the group of $1$-cycles. The Hurewicz theorem states that the map $\pi_1(S,v)\to H_1(S;\mZ)$ that sends (the homotopy class of ) a closed walk to the (homology class of the) formal sum of its oriented edges is onto with kernel the commutator subgroup of $\pi_1(S,v)$. In the case of the torus, $\pi_1(T,v)$ is commutative, so that the above map is actually an isomorphism. From now on we will identify homotopy and first homology classes on the torus. We will denote by the same letter a closed walk on $G$ and the corresponding 1-cycle in $Z_1(G;\mZ)$ equal to the sum of its constituent oriented edges. The homotopy or homology class of a closed walk or 1-cycle $c$ will be indifferently denoted by $[c]$.

\paragraph*{Intersection numbers} Given two closed oriented curves $c,d$ on $S$ (endowed with an orientation) with transverse intersections, we may assign a sign to each intersection according to whether the tangents of $c$ and $d$ at the intersection form a positively oriented basis. The sum of the signs over all intersections is called the \define{algebraic intersection number}. It is a classical result that this number only depends on the homology classes $[c]$ and $[d]$ and that it defines an antisymmetric, nondegenerate bilinear form on $H_1(S;\mZ)$, denoted by the \define{pairing} $\langle [c],[d]\rangle$. Of course the total number of intersections of $c$ and $d$ is at least $|\langle [c],[d]\rangle|$.

\paragraph*{The universal cover of the torus} We can form a torus by identifying the opposite sides of a square. Equivalently, we can see a torus as the quotient space of the plane $\mR^2$ by the action of the group of translations generated by $(1,0)$ and $(0,1)$, which we identify with the lattice $\mZ^2$. Hence, we can identify $T$ with $\mR^2/\mZ^2$ and we have a quotient map $q: \mR^2\to T$. The plane $\mR^2$, with the map $q$, is called the \define{universal cover} of $T$. Given a curve $c$ with source point $v$ on $T$, and a point $\tilde{v}\in q^{-1}(v)$, there is a unique curve $\tilde{c}$ in the plane with source $\tilde{v}$ that \textit{projects} to $c$, i.e., such that $q(\tilde{c})=c$. The curve $\tilde{c}$ is called a \define{lift} of $c$. If $c$ is a closed curve, then the vector from the source to the target of $\tilde{c}$ has integer coordinates and only depends on $[c]$. Hence, each homotopy class can be identified with a lattice translation. Such translations are called \define{covering transformations} (or translations). It can be proved that a curve is freely homotopic to a simple curve if and only if the coordinates of the corresponding covering translation are coprime~\cite[Sec. 6.2.2]{stillwell1993}. By the identification between $\mZ^2$, $\pi_1(T,v)$ and $H_1(T;\mZ)$, any pair $(\alpha,\beta)$ of homology classes that generates $H_1(T;\mZ)$ must correspond to an invertible integer transformation, hence to a unimodular matrix. Equivalently, $(\alpha,\beta)$ has algebraic intersection number equals to $\pm{1}$. It is a \define{positively oriented basis} when $\langle \alpha,\beta \rangle = 1$.

\paragraph*{Integer and intersection norms}
Let $N: \mZ^d\to \mR_{\geq 0}$ satisfy the norm axioms: 
\begin{itemize}
\item $N(\alpha+\beta)\leq N(\alpha)+N(\beta)$ (subadditivity)
\item $N(k\alpha) = |k|N(\alpha)$ (absolute homogeneity) 
  \item $N(\alpha)=0 \implies \alpha=0$ (separation)
\end{itemize}
Then $N$ extends to the rational space $\mQ^d$ using homogeneity, and can be extended to $\mR^d$ so that it is continuous. It can be shown that this indeed provides a well-defined norm over $\mR^d$~\cite{t-nh3m-86}. Such a function $N$, or its real extension, is called an \define{integer norm} if $N(\mZ^d) \subseteq \mZ_{\geq 0}$. It can be proved that integer norms are \define{polyhedral}, i.e. their unit ball is a centrally symmetric polytope, and that their dual unit ball is a centrally symmetric polytope with integer vertices~\cite{t-nh3m-86,schrijver1993,s-ntivi-16}. See also~\cite[Sec. 6.0.4]{c-dslam-20}. 

Integer norms naturally arise as length functions defined over homology classes of curves on surfaces. There are several ways to define curves and their lengths with respect to a graph $G$ embedded on a surface $S$. One can consider continuous curves on $S$ and define their length as the number of crossings with $G$. Schrijver~\cite{schrijver1993} applies this framework when $S$ is a torus and shows that this indeed defines a norm. He also considers a framework where the curves are in general position with respect to $G$, thus avoiding its vertices, and where all the vertices of $G$ are required to have degree 4.
In~\cite{schrijver1992} Schrijver shows that the first framework
reduces to the second by considering the medial graph of $G$. In turn, the second framework reduces to our framework by duality, in the special case where the faces are quadrilaterals and the edges are unweighted.
The second framework is also used by Kane~\cite{s-inofc-20}, where the norm is  referred to as an \define{intersection norm}. It is easily seen that the intersection norm for unweighted 4-regular graphs is equal to 
the marked length spectrum on the dual graph of $G$. In~\cite{s-inofc-20}, examples are given of centrally symmetric polytopes with integer vertices which are not dual balls of any intersection norm. This is in contrast with the special case where $G$ is embedded on the torus, for which Schrijver~\cite[Th. 4]{schrijver1993} shows that every integer norm on $\mR^2$ is (half) the intersection norm of an unweighted 4-regular graph.

\section{Length spectrum and polyhedral norms on homology}\label{sec:polyhedralnorms}
In this section, given a weighted graph $(G,w)$ embedded on a torus $T$, we introduce a norm on the first homology group of the torus that will be used throughout the article. A correspondence between graphs on the torus and polyhedral norms has been known for some time~\cite{schrijver1993}. But, as far as we know, it has been studied only in the unweighted case and furthermore never analyzed from a computational point of view.

For $\alpha \in H_1(T; \mZ)$ let
\begin{equation}
\label{eq:normDefinition}
\Norm{G,w}(\alpha) :=
\inf \left\{ \sum_{e \in E(G)} |x_e| w(e) : \sum_{e \in E(G)} x_e e \in Z_1(G; \mZ) \text{ and } [\sum_{e \in E(G)} x_e e ] = \alpha \right\}.
\end{equation}
Note that the infimum in (\ref{eq:normDefinition}) is attained, because when one of the $x_e$ tends to infinity, $\sum_{e \in E(G)} |x_e| w(e)$ tends to infinity as well, and so our minimization problem can actually be constrained to a finite set.

In the proof of Theorem \ref{thm:weightedgraphspolyhedralnorm} below we show that $\Norm{G,w}$ thus defined satisfies the norm axioms. Hence, as explained in the subsection ``Integer and intersection norms'' of Section \ref{sec:preliminaries}, $\Norm{G,w}$ extends to a norm on $H_1(T;\mR)$ (because $H_1(T;\mZ)$ is naturally a lattice in $H_1(T;\mR$)). 

But first let us show that $\Norm{G,w}$ as defined in (\ref{eq:normDefinition}) is indeed the marked length spectrum of $T$ with respect to $(G,w)$.

\begin{lemma}
\label{lem:normDefinitionSimpleCycles}
For every $\alpha \in H_1(T;\mZ)$ we have
    \begin{equation}
	    \label{eq:normDefinitionSimpleCycles}
	    \Norm{G,w}(\alpha) = \inf \left\{ \sum_{i \in I} x_i \cdot w(c_i) : [\sum_{i \in I} x_i \cdot c_i] = \alpha \text{ and } x_i \in \mZ_{\geq 0} \text{ for } i \in I \right\},
    \end{equation}
where $\{c_i\}_{i \in I}$ is the (finite) set of all simple cycles in $G$. The infimum in (\ref{eq:normDefinitionSimpleCycles}) is attained.
Furthermore, for every $\alpha \in H_1(T;\mZ)$, $\Norm{G,w}(\alpha)$ is the length of a shortest closed walk $c$ in $G$ with $[c]=\alpha$. In other words, $\Norm{G,w}$ is the marked length spectrum of $T$ with respect to $(G,w)$.
\end{lemma}
\begin{proof}
	Denote by $\Norm{G,w}'(\alpha)$ the right-hand side of (\ref{eq:normDefinitionSimpleCycles}). Let $\{x_i\}_{i \in I}$ be such that $x_i \in \mZ_{\geq 0}$ and $\left[\sum_{i \in I} x_i \cdot c_i\right] = \alpha$. If $\sum_{i \in I} x_i \cdot c_i = \sum_{e \in E(G)} x_e e$, then for every $e \in E(G)$ we have $|x_e| = |\sum_{i \in I, e \in c_i} \pm x_i| \le \sum_{i \in I, e \in c_i} x_i$ and so $\sum_{e \in E(G)} |x_e| w(e) \le \sum_{i \in I} x_i w(c_i)$. Hence $\Norm{G,w}(\alpha) \le \Norm{G,w}'(\alpha)$.
	
	We will now prove the opposite inequality. Fix $c=\sum_{e \in E(G)} x_e e \in Z_1(G;\mZ)$. We claim that there exist $\{x_i\}_{i \in I}$ such that $x_i \in \mZ_{\geq 0}$, $\sum_{i \in I} x_i \cdot c_i = c$ and $\sum_{i \in I} x_i w(c_i) = \sum_{e \in E(G)} |x_e| w(e)$. To prove the claim, we can suppose that $x_e \ge 0$ for all $e \in E(G)$ (changing the orientation of edges if necessary). It follows from the definition of $Z_1(G;\mZ)$ that for every vertex of $G$ the sum of the coefficients in $c$ of its incoming edges is equal to the sum of the coefficients in $c$ of its outgoing edges. Pick $e_1 \in E(G)$ with the smallest $x_{e_1} \neq 0$. Suppose we have already picked distinct edges $e_1, \ldots, e_i$. While the target of $e_i$ is not equal to the source of $e_1$, we pick a next edge $e_{i+1}$ arbitrarily among the edges emanating from the target of $e_i$, with non-zero coefficients in $c$ and that have not been picked yet. The property of $Z_1(G;\mZ)$ stated above ensures that at least one such edge exists. Suppose that this procedure terminates at stage $k$. Then the edges $e_1, \ldots, e_k$ form a simple cycle $d_1$ in $G$. Replace now $c$ by $c - x_{e_1} d_1$. This new element of $Z_1(G;\mZ)$ still has non-negative coefficients, and the number of positive coefficients decreased by at least 1. Repeating this procedure, we get the desired representation $c = x_{e_1} d_1 + \ldots$.
	
	The claim above applies in particular to a minimal representation $[\sum_{e \in E(G)} x_e e] = \alpha$ such that $\Norm{G,w}(\alpha)= \sum_{e \in E(G)} |x_e| w(e)$. Hence, $\Norm{G,w}(\alpha) \ge \Norm{G,w}'(\alpha)$, and we get (\ref{eq:normDefinitionSimpleCycles}). This also gives a minimal representation $[\sum_{i \in I} x_i c_i] = \alpha$ such that $\Norm{G,w}(\alpha)= \sum_{i \in I} x_i w(c_i)$, and so the infimum in (\ref{eq:normDefinitionSimpleCycles}) is indeed attained.
	
	Finally, suppose $\sum_{i \in I} x_i c_i$ representing $\alpha$ is such that $\Norm{G,w}(\alpha) = \sum_{i \in I} x_i w(c_i)$. We can assume that all the homology classes $[c_i]$ are different (otherwise replace the simple cycles $c_i$ corresponding to the same class by any one of them, not changing the total weight). Moreover, $[c_i]$ are pairwise non-collinear. Indeed, if $[c_i] = -[c_j]$, we could replace $x_i c_i + x_j c_j$ either by $(x_i-x_j) c_i$ or $(x_j-x_i) c_j$ thus decreasing the total weight. If $[c_i] = \pm p/q \cdot [c_j]$ with $p,q \in \mZ_{>0}$ relatively prime and $(p,q)\neq (1,1)$, then one of the vectors $[c_i], [c_j] \in H_1(T; \mZ)$ is not primitive (i.e. it is an integer multiple of other class from $H_1(T; \mZ)$) and so (as noted in Section \ref{sec:preliminaries}) there cannot be a simple closed curve with this homology class, a contradiction.
	
	Since on the torus any two non-collinear homology classes intersect, any two simple cycles $c_i$ and $c_j$ in the minimal representation of $\alpha$ intersect, and so they all can be concatenated to form a closed walk on $G$ with homology $\alpha$ and of length $\Norm{G,w}(\alpha)$.
\end{proof}

In the proof of Theorem \ref{thm:weightedgraphspolyhedralnorm} we show that the extremal points of the unit ball $\Ball{G,w} = \{\alpha \in H_1(T; \mR) \mid \Norm{G,w}(\alpha) \le 1\}$ of $\Norm{G,w}$ correspond to homology classes that can be represented by simple cycles in $G$. The following lemma gives an upper bound on the number of such homology classes.

For a subset $X$ of a real vector space let $\conv(X)$ denote the convex hull of $X$. Note that $H_1(T;\mZ)$ is naturally a subset of the real vector space $H_1(T;\mR$).

\begin{lemma}\label{lem:numberofsimplecycles}
	Let $G$ be a graph with $|V|$ vertices cellularly embedded on the torus $T$, and let $\SimpleC \subset H_1(T;\mZ)$ be the set of homology classes of curves that can be represented as simple cycles in $G$. Then in $H_1(T;\mR)$ we have $|\conv(\SimpleC) \cap H_1(T;\mZ)| \le 4|V|+5$.
\end{lemma}

\begin{proof}
Identify $H_1(T;\mZ)$ with $\mZ^2$ via an arbitrary positively oriented basis. $H_1(T;\mR)$ is then identified with $\mR^2$, and the algebraic intersection pairing is given by $\langle (x,y), (x',y') \rangle = xy' - x'y$. Note that the absolute value of this pairing is the Euclidean area of the parallelogram spanned by these vectors.
	
    Let $\alpha,\beta \in \SimpleC$ be represented by simple cycles $c_\alpha, c_\beta$ in $G$. On the one hand, the number of intersections between $c_\alpha$ and $c_\beta$ is bounded by $|V|$, since each intersection corresponds to at least one vertex of $G$, and all these vertices must be different. On the other hand, it is bounded below by the algebraic intersection number $|\langle \alpha,\beta \rangle|$. Hence, $|\langle \alpha,\beta \rangle| \leq |V|$.
	
	Denote by $\lVert \cdot \rVert$ the Euclidean norm on $\mR^2$ and by $\dist(\cdot,\cdot)$ the Euclidean distance, and consider $\alpha, \beta \in \SimpleC$ such that $|\langle \alpha, \beta \rangle|$ is maximal. Then for any $\gamma \in \SimpleC$, we have $|\langle \gamma, \alpha \rangle| \le |\langle \alpha, \beta \rangle|$ and $|\langle \gamma, \beta \rangle| \le |\langle \alpha, \beta \rangle|$. 
	Note that, since these numbers are the areas of the corresponding parallelograms, $|\langle \gamma, \alpha \rangle| = \lVert \alpha \rVert \cdot \dist(\gamma, \mR \alpha)$, $|\langle \gamma, \beta \rangle| = \lVert \beta \rVert \cdot \dist(\gamma, \mR \beta)$ and $|\langle \alpha, \beta \rangle| = \lVert \alpha \rVert \cdot \dist(\beta, \mR \alpha)=\lVert \beta \rVert \cdot \dist(\alpha, \mR \beta)$, where $\mR\alpha, \mR\beta$ denote the one-dimensional $\mR$-subspaces generated by $\alpha, \beta$ respectively. It follows that $\dist(\gamma, \mR \alpha) \le \dist(\beta, \mR \alpha)$ and $\dist(\gamma, \mR \beta) \le \dist(\alpha, \mR \beta)$, and so $\SimpleC$ is contained in the parallelogram $P$ with vertices $\pm \alpha \pm \beta$, see Figure~\ref{fig:SCG}.
          \begin{figure}[h]
            \centering
            \includesvg[.4\textwidth]{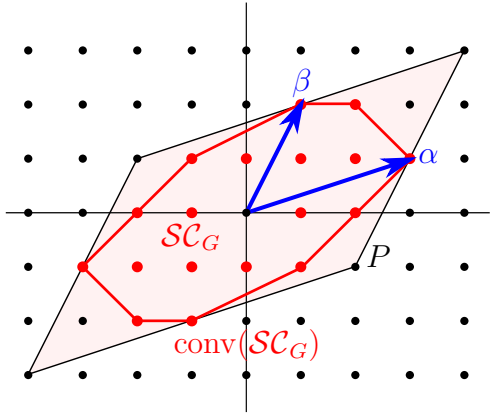}
    \caption{The elements of $\SimpleC$ are contained in a parallelogram $P$ of area at most $4|V|$.}
    \label{fig:SCG}
    \end{figure}
    
    Clearly, the area $A(P)$ of  $P$ is $4|\langle \alpha, \beta \rangle|$. At the same time, by Pick's theorem $A(P) = I+B/2-1$, where $I$ is the number of integer points strictly inside $P$ and $B$ is the number of integer points on its boundary. Since $\alpha$ and $\beta$ are homology classes represented by simple cycles, their corresponding vectors in $\mZ^2$ have coprime coordinates, i.e. the only integer points on the vectors $\alpha$ and $\beta$ are their endpoints. It follows that the only integer points on the boundary of $P$ are $\pm \alpha, \pm \beta, \pm \alpha \pm \beta$ and so $B=8$.
    
    Finally, since $\SimpleC \subset P$, we have $\conv(\SimpleC) \subset P$ as well, and so
	\[
	    |\conv(\SimpleC) \cap \mZ^2| \le  I+B = A(P) + B/2 + 1 = 4|\langle \alpha, \beta \rangle| + 5 \le 4|V|+5.
    \]
\end{proof}

Note that the linearity of the bound in Lemma \ref{lem:numberofsimplecycles} is optimal. 
Indeed, there are graphs $G$ embedded on the torus with $|V|$ vertices satisfying $|\SimpleC|=\Theta(|V|)$. One can take for example a graph obtained by drawing on the torus $\mR^2 / \mZ^2$ the curves $q(\mR \times \{1/2\})$ and $q(\mR \cdot (1,k))$, where $q:\mR^2 \rightarrow \mR^2 / \mZ^2$ is the universal covering map, and putting vertices at the intersections of these curves. See Figure~\ref{fig:SCG_example} for an example with $k=4$. The resulting graph $G$ has $k$ vertices and $\SimpleC = \{\pm (0,1), \pm (1,0), \pm (1,1), \pm (1,2), \ldots, \pm (1,k)\}$ with $|\SimpleC|=2k+4$.

	\begin{figure}[h]
    \centering
    \includesvg[.2\textwidth]{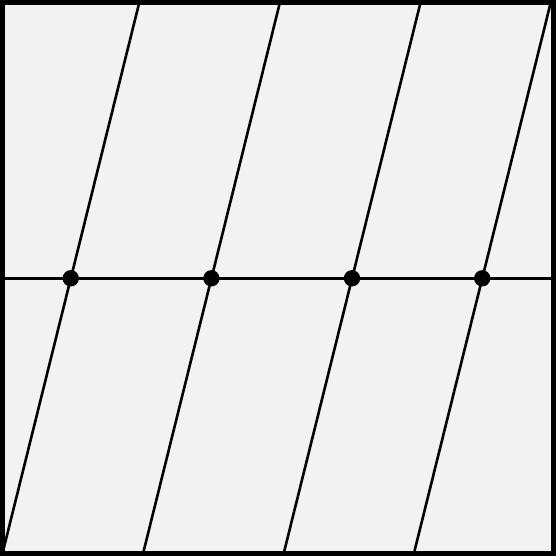}
    \caption{An example of a graph $G$ with $k=4$ vertices and $|\SimpleC| = 2k+4$.}
    \label{fig:SCG_example}
    \end{figure}
    
We now pass to the proof of Theorem \ref{thm:weightedgraphspolyhedralnorm}. In the unweighted case, the polyhedrality of the norm is a consequence of the integrality, see~\cite{schrijver1993}. However, this argument does not apply in the weighted case.

\begin{proof}[Proof of Theorem~\ref{thm:weightedgraphspolyhedralnorm}]
	It is clear from the definition of $\Norm{G,w}$ that for all $\alpha, \beta \in H_1(T; \mZ)$ we have $\Norm{G,w}(\alpha+\beta) \le \Norm{G,w}(\alpha) + \Norm{G,w}(\beta)$ and that $\Norm{G,w}(\alpha) = 0$ implies $\alpha = 0$. To prove absolute homogeneity, we use the characterisation of $\Norm{G,w}$ as a marked length spectrum of $T$ with respect to $(G,w)$, proven in Lemma \ref{lem:normDefinitionSimpleCycles}. First note that if $c$ is the shortest closed walk in $G$ with $[c]=\alpha$, then $c$ with opposite orientation is the shortest closed walk in $G$ with homology $-\alpha$. Hence, $\Norm{G,w}(-\alpha)=\Norm{G,w}(\alpha)$. 
	
	Now, we prove by induction on $k \in \mZ_{>0}$ that $\Norm{G,w}(k\alpha)=k \cdot \Norm{G,w}(\alpha)$. For $k=1$ this is obviously true. For $k\ge 2$, the class $k\alpha \in H_1(T;\mZ)$ is not primitive, and so the shortest closed walk $c$ in $G$ representing $k\alpha$ necessarily self-intersects (as noted in Section \ref{sec:preliminaries}). It follows easily from the proof in \cite[Sec. 6.2.2]{stillwell1993} that the point of self-intersection can be chosen in such a way that it decomposes $c$ into two closed walks $c_1$ and $c_2$ with $[c_1]=k_1 \alpha, [c_2]=k_2 \alpha$, $k_1+k_2=k$. If $k_1=0$ or $k_2=0$, then one of the walks can be removed from $c$ not changing the homology and decreasing length, which is impossible. So $k_1, k_2 >0$ and $\Norm{G,w}(k\alpha) = w(c) \ge \Norm{G,w}(k_1\alpha) + \Norm{G,w}(k_2\alpha) = k\Norm{G,w}(\alpha)$ by induction hypothesis. The opposite inequality follows from the triangle inequality.
	
	So $\Norm{G,w}$ satisfies the norm axioms. Hence, as explained in the subsection ``Integer and intersection norms'' of Section \ref{sec:preliminaries}, $\Norm{G,w}$ extends to a norm on $H_1(T;\mR)$. To prove the polyhedrality of this norm, we show that \begin{equation}
	\label{eq:ballConvexHull}
	    \Ball{G,w} = \conv \left( \left\{ \frac{[c_i]}{w(c_i)} \mid i \in I \right\} \right),
	\end{equation}
	where $c_i$ is the (finite) set of all oriented simple cycles in $G$, as in Lemma \ref{lem:normDefinitionSimpleCycles}.
	
	Denote the right-hand side of (\ref{eq:ballConvexHull}) by $\Ball{G,w}'$. Clearly, for every $i \in I$ we have $\Norm{G,w}([c_i])\leq w(c_i)$, so $\Ball{G,w}' \subset \Ball{G,w}$. Conversely, take any homology class $\alpha \in H_1(T;\mZ)$. By Lemma \ref{lem:normDefinitionSimpleCycles}, we have $\Norm{G,w}(\alpha) = \sum_{i \in I} x_i \cdot w(c_i)$ for some $x_i \in \mZ_{\ge 0}$ such that $\alpha = [\sum_{i \in I} x_i \cdot c_i]$. Then 
	\[\frac{\alpha}{\Norm{G,w}(\alpha)} = \frac{\sum_{i \in I} x_i \cdot [c_i]}{\sum_{i \in I} x_i \cdot w(c_i)} = \sum_{i \in I} \frac{x_i w(c_i)}{\sum_{j \in I} x_j \cdot w(c_j)} \cdot \frac{[c_i]}{w(c_i)}\]
	is a representation of $\frac{\alpha}{\Norm{G,w}(\alpha)}$ as a convex combination of $\frac{[c_i]}{w(c_i)}, i \in I$. Hence $\Ball{G,w} \subset \Ball{G,w}'$ and we get (\ref{eq:ballConvexHull}).
	
	By definition, the homology classes $[c_i], i \in I$ can be represented as simple cycles in $G$. By Lemma \ref{lem:numberofsimplecycles} their number is at most $4|V|+5$, and so the number of extremal points of $\Ball{G,w}$ is also at most $4|V|+5$. The slopes of $[c_i]/w(c_i)$ are rational since the $[c_i]$ belong to $H_1(T;\mZ)$.
\end{proof}

Finally, we show how to reconstruct a weighted graph $(G,w)$ embedded on the torus $T$ from a polyhedral norm on $\mR^2$. The case of integral norms and unweighted graphs is treated in~\cite{schrijver1993}. Our method is different and does not provide an optimal construction (in the number of vertices) contrarily to~\cite{schrijver1993}. This minimization problem is probably hard in the weighted case.

\begin{theorem} \label{thm:normReconstruction}
Let $N: \mR^2 \to \mR$ be a polyhedral norm all of whose extremal points have rational
slopes. Let $\{\pm (p_i,q_i)\}_{i=1,\ldots,n}$ be the set of non-zero integral vectors closest to the origin on the rays issued from the origin in the direction of the extremal points of the unit ball $\{v \in \mR^2 : N(v) \le 1\}$. Then there exists
a weighted 4-valent graph $(G,w)$ embedded on the torus $T$ with
$\displaystyle \sum_{1 \le i < j \le n} |p_i q_j - p_j q_i|$ vertices so that $\Norm{(G,w)} = N$.
\end{theorem}

\begin{proof}
    Let $D_i$ be the projection of the segment $[(x_i, y_i), (x_i+p_i, y_i+q_i)]$ on the torus $T = \mR^2 / \mZ^2$ where $(x_i, y_i) \in \mR^2$ are chosen so that at most two curves $D_i$ pass through each point of $T$. The union $G$ of the curves $D_i$ is a graph embedded
    in $T$ with $\displaystyle \sum_{1 \le i < j \le n} |p_i q_j - p_j q_i|$ 4-valent vertices, one for each intersection between $D_i$ and $D_j$. To each edge $e$ of $G$ we associate the weight $w(e) = N(v_e)$ where $v_e$ is the vector in $\mR^2$ between the two endpoints of a lift of $e$ to $\mR^2$. 
    
    By construction we have $\Norm{G,w}((p_i,q_i)) \le N((p_i,q_i))$ since the class $(p_i,q_i)$ is realized by the path corresponding to the projection of $D_i$ which has weight $N((p_i, q_i))$. For any two adjacent extremal points $v$ and $v'$ of the unit ball of $N$ we thus have $N(tv + (1-t)v') = tN(v) + (1-t)N(v')$, and so $\Norm{G,w}(tv + (1-t)v') \le t\Norm{G,w}(v) + (1-t)\Norm{G,w}(v') \le tN(v) + (1-t)N(v') = N(tv + (1-t)v')$. Hence $\Norm{G,w} \le N$.
    
    Our aim is now to show that $\Norm{G,w} \ge N$. It is enough to show that it is true for integral points, that is elements in $H_1(T; \mZ)$. Let $\alpha = \sum_{e \in E(G)} x_e e \in Z_1(G;\mZ)$ such that $\Norm{G,w}([\alpha])=w(\alpha).$ Then $\Norm{G,w}([\alpha]) = \sum |x_e| w(e) = \sum |x_e| N(v_e)$ where the last equality is our choice of weights. By construction the vector $\sum_{e \in E(G)} x_e v_e \in \mR^2$ is equal to $[\alpha]$ if we identify $H_1(T; \mR)$ with $\mR^2$. Hence $\Norm{G,w}([\alpha]) = \sum |x_e| N(v_e) \ge N(\sum x_e v_e) = N([\alpha])$ where the inequality is the triangular inequality for $N$. Hence $\Norm{G,w} \ge N$ and so the norms $\Norm{G,w}$ and $N$ coincide.
\end{proof}

\section{Good short basis}\label{sec:short-basis}
Our computation of the length spectrum and of its unit ball relies on the initial computation of a \define{good short basis}. By a \define{short basis} we mean a pair of tight simple cycles $(a,b)$ in $G$ such that $a$ is a shortest non-trivial closed walk and $b$ is a shortest non-trivial closed walk satisfying $\langle [a],[b] \rangle = 1$. We say that $(a,b)$ is a \emph{good basis} if $([a],[b])$ is a positively oriented basis of $\HoneT$ and  $a$ and $b$ intersect along a connected path, possibly reduced to a vertex.

\begin{lemma}\label{lem:short-basis}
  Let $(G,w)$ be a weighted graph of complexity $n$ cellularly embedded on the torus. A good short basis can be computed in $O(n\log n)$ time.
\end{lemma}
\begin{proof}
  We first compute a shortest non-trivial closed walk $a$ in $O(n\log n)$ time following Kutz~\cite[Th. 1]{k-csntc-06}. This closed walk must be a tight simple cycle as otherwise it could be decomposed into shorter non-trivial closed walks. We define an \define{$a$-line} as any bi-infinite concatenation of lifts of $a$ in the universal cover $q: \mR^2\to T$. Since $a$ is tight, every $a$-line is geodesic, meaning that it includes a shortest path in the lifted graph $q^{-1}(G)$ between any pair of its vertices.

  We claim that among all shortest non-trivial closed walks $b$ satisfying $\langle [a],[b] \rangle = 1$ there is one that intersects $a$ along a connected path. Indeed, let  $b'$ be a shortest non-trivial closed walk satisfying $\langle [a],[b'] \rangle = 1$. Let $v$ be a vertex in $b'\cap a$ and let $\tilde{v}\in q^{-1}(v)$. Denote by $\tilde{w}$  the target vertex of the lift $\tilde{b'}$ of $b'$ from $\tilde{v}$   in the universal cover. Since the algebraic intersection number of $a$ and $b'$ is one, the vertices $\tilde{v}$ and $\tilde{w}$ join two consecutive $a$-lines, say $L$ and $L+\tau_{b'}$, where $\tau_{b'}$ is the covering translation associated to $b'$. See Figure~\ref{fig:short-basis}.
\begin{figure}[h]
	\centering
	\includesvg[\textwidth]{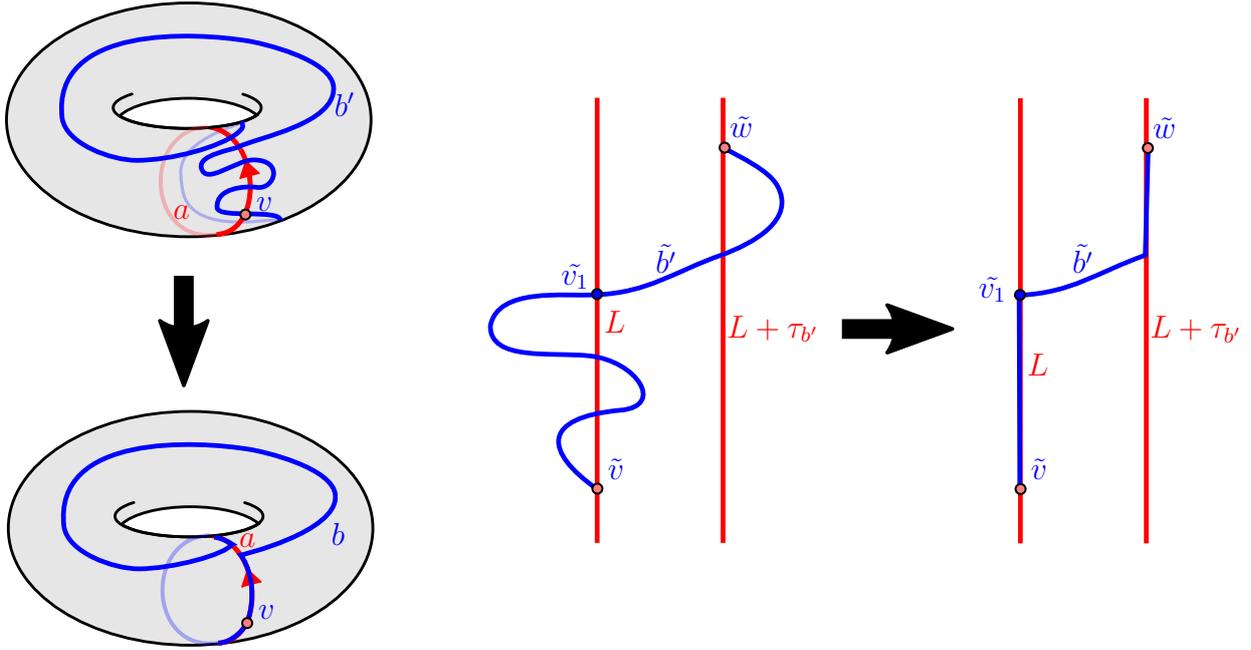}
	\caption{Left: modifying $b'$ so that it intersects $a$ in a connected path. Right: the lift of $b'$ can be enforced to lie between two consecutive $a$-lines.}
	\label{fig:short-basis}
\end{figure}

  Consider the last intersection $\tilde{v_1}$ of $\tilde{b'}$ with $L$. We replace in $b'$ the subpath between $\tilde{v}$ and $\tilde{v_1}$ by a non-longer path along $L$. We next replace the subpath of the resulting path between $\tilde{w}$ and the first intersection  with  $L+\tau_{b'}$ by a path along $L+\tau_{b'}$. This yields a path $\tilde{b}$ so that $b:=q(\tilde{b})$ satisfies the claim.
 
   Cutting $T$ along $a$ yields an annulus $A$ with two copies $a'$ and $a''$ of $a$ as boundary components. By the above claim,  $b$ intersects $A$ in a shortest path connecting two copies of the same vertex respectively on $a'$ and $a''$. To find this shortest path, we first use the multiple-source shortest path algorithm of Klein for sources on the outer face of a plane graph; see~\cite{k-msspp-05} and~\cite[Th. 3.8]{cce-msspe-13}. This algorithm builds a data structure in $O(n\log n)$ time that allows to query for the distance between any vertex on $a'$ and any other vertex in $A$ in $O(\log n)$ time. We need to query for the $O(n)$ pairs of copies of vertices of $a$ and retain a pair $(u',u'')$ that minimizes the distance. In order to find an explicit representative of $b$, we can in a second step run Dijkstra's algorithm with source $u'$ in $A$. Finally, to ensure that $b$ intersects $a$ along a connected path, we can replace the subpath between $u'$ and the last occurrence of a vertex on $a'$ by a subpath of $a'$ with the same length and do similarly on $a''$. The total running time is $O(n\log n)$. We obtain $b$ by gluing back the two copies of $a$. Note that the resulting closed walk must be simple. Indeed, since $b$ is computed from a shortest path in $A$, the only vertices that may appear twice are copies of a same vertex on $a'$ and $a''$. This would however yield a shorter path between corresponding vertices, in contradiction with the minimizing property of $(u',u'')$.
\end{proof}

We shall always express a homology class in the basis $(a, b)$ and identify the class with a vector in $\mZ^2$. Hence, the homology classes of $a$ and $b$ are identified with $(1,0)$ and $(0,1)$, respectively.

\section{Computing the unit ball}
\label{sec:computingUnitBall}

Here, we provide an algorithm for computing the unit ball $\Ball{G,w}$ of the norm $\Norm{G,w}$ corresponding to the weighted graph $(G,w)$. Let $\TSimpleC \subset H_1(T;\mZ)$ be the set of homology classes that admit a tight and simple representative in $G$. Of course, $\TSimpleC \subseteq \SimpleC$, and the homology classes of $a$ and $b$ computed in Section~\ref{sec:short-basis} are in $\TSimpleC$ by construction. 

In Section~\ref{sec:polyhedralnorms}, we proved that $\Ball{G,w}$ is the convex hull of a set $\{\alpha/\Norm{G,w}(\alpha)\}_{\alpha \in \SimpleC}$ containing $O(|V|)$ classes.
We shall compute a subset $H$ of $\TSimpleC$ whose normalized vectors, $\{{\alpha}/{\Norm{G,w}(\alpha)}\}_{\alpha\in H}$, include all the extremal points of $\Ball{G,w}$. Since the coordinates of each element of $\TSimpleC$ must be coprime, the set of directions defined by the elements of $\TSimpleC$ are pairwise distinct and naturally ordered angularly. We search for $H$ by exploring the whole set of directions using dichotomy together with a simple pruning strategy. Suppose we need to explore the angular sector $\angle(\alpha,\beta)$, where $(\alpha,\beta)$ forms a basis of $\HoneT$. The dichotomy consists in splitting the sector into the sectors $\angle(\alpha,\gamma)$ and $\angle(\gamma,\beta)$ with $\gamma:=\alpha+\beta$. Note that $(\alpha,\gamma)$ and $(\gamma,\beta)$ are again bases of $\HoneT$. In particular, the coordinates of $\gamma$ are coprime\footnote{$\gamma$ corresponds to the \emph{mediant} of $\alpha,\beta$ that appears in Farey sequences.}. Since for any nonzero $\eta\in \HoneT$, the normalized vector ${\eta}/{\Norm{G,w}(\eta)}$ lies on the boundary of the unit ball, it follows by convexity of $\Ball{G,w}$ that the segment $[\frac{\alpha}{\Norm{G,w}(\alpha)}, \frac{\beta}{\Norm{G,w}(\beta)}]$ is a subset of a supporting line of $\Ball{G,w}$ whenever ${\gamma}/{\Norm{G,w}(\gamma)}$ lies on this segment. This last condition has a simple certificate.
\begin{claim}\label{clm:prune}
  $\frac{\gamma}{\Norm{G,w}(\gamma)}$ lies on the segment $[\frac{\alpha}{\Norm{G,w}(\alpha)}, \frac{\beta}{\Norm{G,w}(\beta)}]$ if and only if\\ $\Norm{G,w}(\alpha+\beta)=\Norm{G,w}(\alpha)+\Norm{G,w}(\beta)$.
\end{claim}
\begin{proof}
  If $\Norm{G,w}(\alpha+\beta)=\Norm{G,w}(\alpha)+\Norm{G,w}(\beta)$, then
\[\frac{\gamma}{\Norm{G,w}(\gamma)}= \frac{\alpha}{\Norm{G,w}(\alpha+\beta)} + \frac{\beta}{\Norm{G,w}(\alpha+\beta)} = \frac{\Norm{G,w}(\alpha)}{\Norm{G,w}(\alpha+\beta)} \frac{\alpha}{\Norm{G,w}(\alpha)} + \frac{\Norm{G,w}(\beta)}{\Norm{G,w}(\alpha+\beta)} \frac{\beta}{\Norm{G,w}(\beta)},
  \]
  and the two positive coefficients in front of $\frac{\alpha}{\Norm{G,w}(\alpha)}$ and $\frac{\beta}{\Norm{G,w}(\beta)}$ add up to one. Conversely, if $\frac{\gamma}{\Norm{G,w}(\gamma)}= t\frac{\alpha}{\Norm{G,w}(\alpha)}+(1-t)\frac{\beta}{\Norm{G,w}(\beta)}$ for some $t\in(0,1)$, then the uniqueness of the decomposition of a homology class in the $(\alpha,\beta)$ basis implies $t = \frac{\Norm{G,w}(\alpha)}{\Norm{G,w}(\alpha+\beta)}$ and $1-t = \frac{\Norm{G,w}(\beta)}{\Norm{G,w}(\alpha+\beta)}$, whence $\Norm{G,w}(\alpha+\beta)=\Norm{G,w}(\alpha)+\Norm{G,w}(\beta)$.
\end{proof}
It follows from the previous discussion that if $\Norm{G,w}(\alpha+\beta)=\Norm{G,w}(\alpha)+\Norm{G,w}(\beta)$, then the interior of the sector $\angle(\alpha,\beta)$ cannot contain any extremal point and we can prune this whole sector in our search. This leads to the pseudo-code of Algorithm~\ref{alg:algo1} for computing $H$. In the sequel, we say that a pair of closed walks in $G$ is \textbf{good} if they are simple and tight cycles, if their homology classes form a basis of $\HoneT$, and if they moreover intersect along a connected path, possibly reduced to a vertex.

  \begin{algorithm}
\caption{Compute $H$}\label{alg:algo1}
\begin{algorithmic}[1]
  \REQUIRE A weighted graph $(G,w)$ cellularly embedded on the torus
  \ENSURE A short basis $(a, b)$ and a sorted list $H = [((x_i,y_i), c_i, w(c_i))]$ where $c_i$ is a simple tight cycle in $G$, $(x_i, y_i) \in \mZ^2$ represents its homology class $[c_i] = x_i [a] + y_i [b]$, and $w(c_i) = \Norm{G,w}([c_i])$. 
  Also, the extremal points of $\Ball{G,w}$ are contained in the set of vectors $\{[c_i] / w(c_i) : i \in \{ 0, \ldots, size(H)-1\} \}$.
  \STATE Compute a good short basis $(a, b)$ as explained in Section~\ref{sec:short-basis}
  \STATE $h_1 := ((1,0), a, w(a))$ \hfill\COMMENT{\small Note that $\Norm{G,w}([a])=w(a)$}
  \STATE $h_2:= ((0,1), b, w(b))$ \hfill\COMMENT{\small and that $\Norm{G,w}([b])=w(b)$.}
  \STATE $\overline{h_1} := ((-1,0), \overline{a}, w(a))$
  \STATE $H:= \{h_1,h_2\}$ \label{alg:initH} \hfill\COMMENT{\small Initialise $H$.}
  \STATE $S:= \{\angle{(h_1,h_2)},\angle{(h_2,\overline{h_1})}\}$ \hfill\COMMENT{\small Initialise a set of sectors to explore with the two upper quadrants.}
  \WHILE{$S\neq \emptyset$} \label{alg:while}
  \STATE Extract and remove from $S$ a sector $\angle{(h, h')}$ \label{alg:removeS}
  \STATE $(x,y), c, \ell := h$ \hfill\COMMENT{Note that $\Norm{G,w}([c])=\ell$.}
  \STATE $(x',y'), c', \ell' := h'$  \hfill\COMMENT{Similarly $\Norm{G,w}([c']) = \ell'$.}
  \renewcommand{\algorithmicrequire}{\textbf{Require:}}
  \REQUIRE $(c, c')$ is a good pair
  \STATE Compute a tight representative $c''$ of $\gamma'' := [c] + [c']$ with its norm $\ell'' := \Norm{G,w}(\gamma'')=w(c'')$ \label{alg:tight}
  \IF{$\ell'' < \ell + \ell'$  \label{alg:if}}
  \STATE $h'':= ((x+x', y+y'), c'', \ell'')$
  \STATE Insert $h''$ in $H$ between $h$ and $h'$ \label{alg:addH}
  \STATE $S :=S \cup \{\angle{(h,h'')}, \angle{(h'',h')}\}$ \label{alg:addS}
 \ENDIF
 \ENDWHILE \label{alg:endwhile}
 \STATE $H:= H\cup \overline{H}$ \label{alg:symmetrize} \hfill\COMMENT{Add the symmetric of $H$ w.r.t. the origin.}
\end{algorithmic}
\end{algorithm}

By subadditivity of the norm, the test in Line~\ref{alg:if} of Algorithm~\ref{alg:algo1} may only fail when $\Norm{G,w}([c'']) = \Norm{G,w}([c]) + \Norm{G,w}([c'])$. It then follows from Claim~\ref{clm:prune} and the preceding discussion on our pruning strategy that we are not missing any direction of extremal points in the upper plane when adding homology classes in Line~\ref{alg:addH}. Moreover, Line~\ref{alg:symmetrize} and
the fact that the unit ball is centrally symmetric ensure that the above algorithm indeed computes a set of homology classes whose normalized vectors contains the extremal points of $\Ball{G,w}$. It remains to explain how to perform the computation in Line~\ref{alg:tight} and to analyse the complexity of Algorithm~\ref{alg:algo1}.

\begin{lemma}\label{lem:addSimple}
  Let $(\alpha,\beta)$ be a homology basis such that $\alpha$ and $\beta$ have representatives, respectively $c_\alpha$ and $c_\beta$, forming a good pair. We can compute a tight representative $c_{\alpha+\beta}$ for $\alpha+\beta$ in $O(n\log\log n)$ time. This reduces to $O(n)$ in the unweighted case.
\end{lemma}
\begin{proof}
  By hypothesis, $c_\alpha$ and $c_\beta$ intersect along a connected path $p_{\alpha\beta}$. There are two cases to consider according to whether $p_{\alpha\beta}$ is oriented the same way or not in  $c_\alpha$ and $c_\beta$. Note that $p_{\alpha\beta}$ may be reduced to a vertex.
  \begin{itemize}
  \item Assume $p_{\alpha\beta}$ is oriented the same way in  $c_\alpha$ and $c_\beta$. We can thus write $c_\alpha = p_\alpha\cdot  p_{\alpha\beta}$ and $c_\beta = p_\beta\cdot  p_{\alpha\beta}$ for some paths $p_\alpha,  p_\beta$ in $G$.  We cut $T$ along the union $c_\alpha\cup c_\beta$, viewed as a subgraph of $G$. We obtain a hexagonal plane domain $\mathcal D$ with sides $p_\beta$, $p_{\alpha\beta}$, $p_\alpha$, $\overline{p_\beta}$, $\overline{p_{\alpha\beta}}$, $\overline{p_\alpha}$ in the clockwise order around the boundary of $\mathcal D$. See Figure~\ref{fig:cutting}.
        \begin{figure}[h]
      \centering
      \includesvg[.75\textwidth]{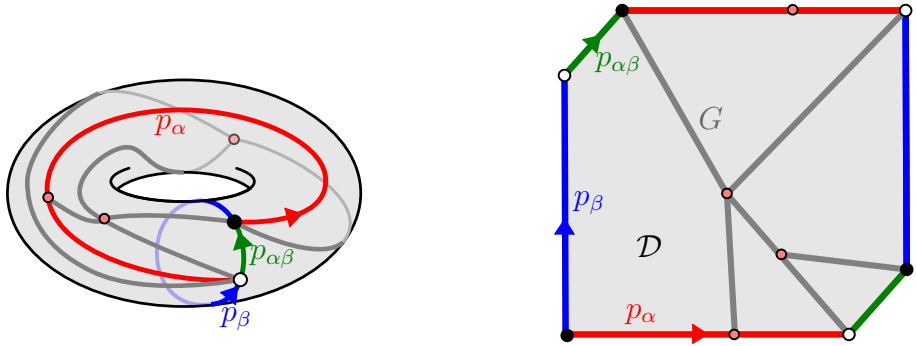}
      \caption{Cutting $T$ along $c_\alpha\cup c_\beta$.}
      \label{fig:cutting}
    \end{figure}
    The universal cover of $T$ is tessellated by translated copies of $\mathcal D$ glued along their sides so that the side $p$ of a domain is glued to the side $\overline{p}$ of the adjacent domain. See Figure~\ref{fig:uniCover}.
    \begin{figure}[h]
      \centering
      \includesvg[\textwidth]{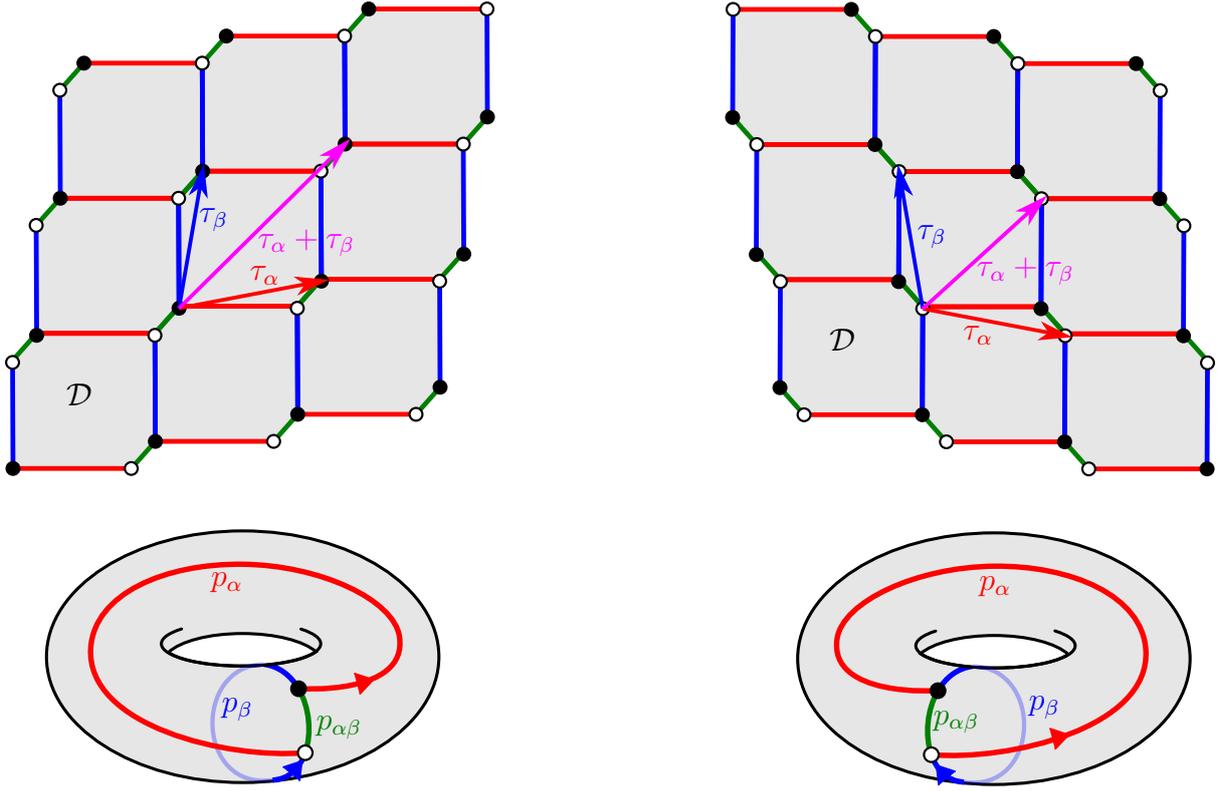}
      \caption{The universal cover of $T$. Left: $p_{\alpha\beta}$ is oriented consistently  with both $c_\alpha$ and $c_\beta$. Right: $p_{\alpha\beta}$  has opposite orientation in $c_\alpha$ and $c_\beta$.}
      \label{fig:uniCover}
    \end{figure}

    As before, let $\gamma=\alpha+\beta$. Since $(\alpha,\gamma)$ is a positively oriented homology basis, we know that $\langle\alpha,\gamma\rangle=1$. Hence, any representative of $\gamma$ must cross $c_\alpha$.  Let $c_\gamma$ be a tight representative of $\gamma$ with a lift $\tilde{c}_\gamma$ in the universal cover starting from a vertex $\tilde{v}$ on the side $p_\alpha$ or $p_{\alpha\beta}$ of a domain $\mathcal{D}_0$ and ending at the vertex $\tilde{w}:=\tilde{v}+\tau_\alpha+\tau_\beta$, where $\tau_\alpha$ and $\tau_\beta$ are the covering translations corresponding to $\alpha$ and $\beta$ respectively. There are two situations according to whether $\tilde{v}$ lies on the side $p_\alpha$ or $p_{\alpha\beta}$ of $\mathcal{D}_0$. See Figure~\ref{fig:D0D1}.
    \begin{figure}[h]
      \centering
      \includesvg[\textwidth]{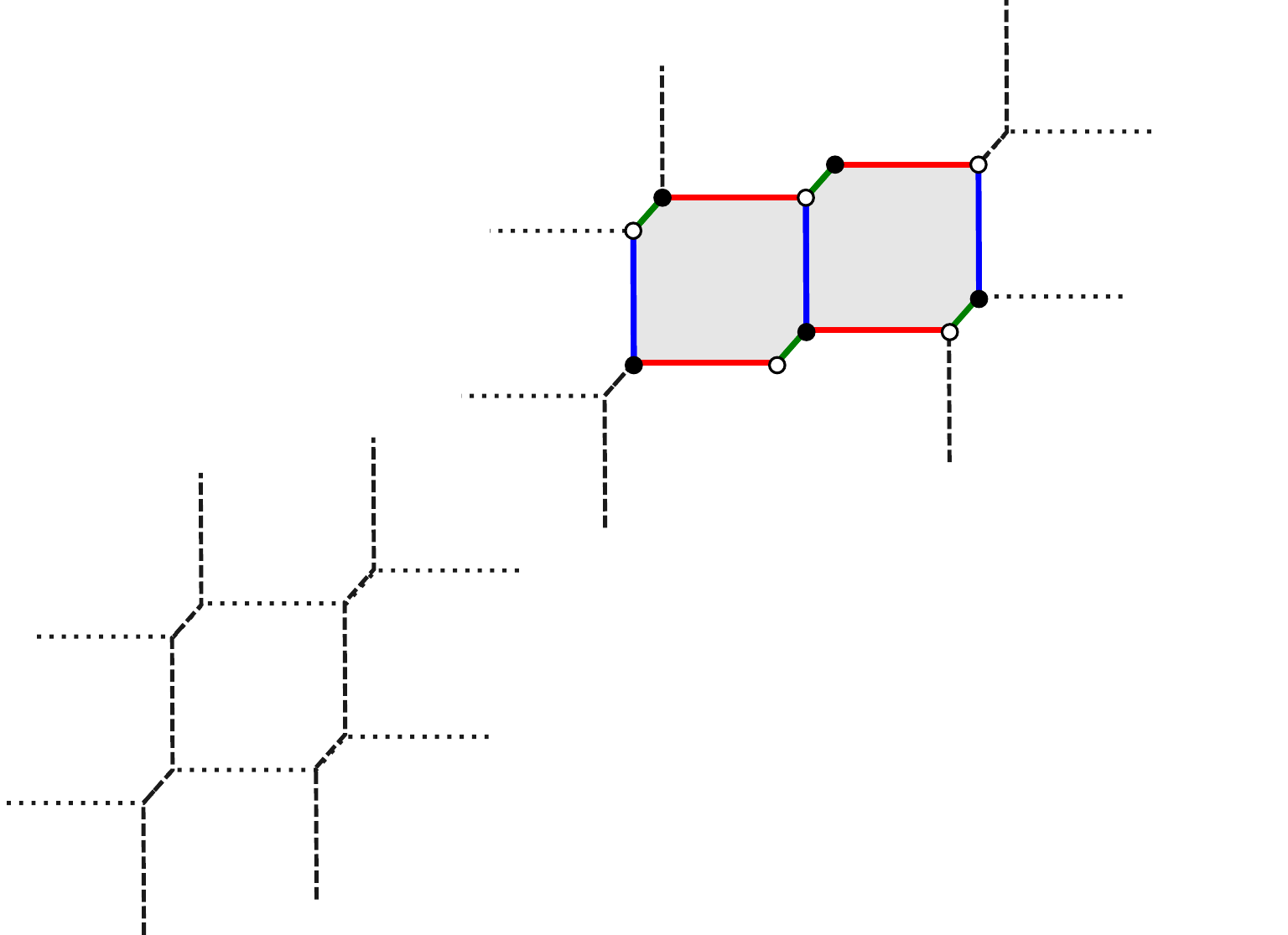}
      \caption{Top: if $c_\gamma$ crosses $p_\alpha$, then $c_\gamma$ has a lift in $\mathcal{D}_0\cup \mathcal{D}_1$. The dotted and broken lines are supporting geodesics for the considered regions. Bottom: when $c_\gamma$ crosses $p_{\alpha\beta}$ it has a lift contained in the union of $\mathcal{D}_1$ with two lifts of  $p_{\alpha\beta}$. We can shift the origin of the lift to $\tilde{v}_0$. }
      \label{fig:D0D1}
    \end{figure}
    \begin{itemize}
    \item If $\tilde{v}$ lies on the side $p_\alpha$ of $\mathcal{D}_0$ then $\tilde{w}$ belongs to the side $\overline{p_\alpha}$ of $\mathcal{D}_1:=\mathcal{D}_0 + \tau_\alpha$. We claim that $\mathcal{D}_0\cup \mathcal{D}_1$ is convex, \ie, any two vertices in $\mathcal{D}_0\cup \mathcal{D}_1$ can be joined by a shortest path contained in  $\mathcal{D}_0\cup \mathcal{D}_1$. Indeed, since $c_\alpha$ is tight, any bi-infinite concatenation of its lifts is a geodesic line in the weighted lift of $G$ in the universal cover of $T$. Similarly,  any bi-infinite concatenation of lifts of $c_\beta$ is a geodesic line  and thus delimits two convex half-planes. See the dotted and broken lines in Figure~\ref{fig:D0D1}. It follows that $\mathcal{D}_0\cup \mathcal{D}_1$ is the intersection of four half-planes, hence is convex. We can thus assume that $c_\gamma$ has a lift in $\mathcal{D}_0\cup \mathcal{D}_1$ with endpoints  $\tilde{v}$ and $\tilde{w}$ on the boundary of $\mathcal{D}_0\cup \mathcal{D}_1$. Let us first describe a method relying on the multiple-source shortest path algorithm of Klein for sources on the outer face of a plane graph (see~\cite{k-msspp-05} and~\cite[Th. 3.8]{cce-msspe-13}) as in the proof of Lemma~\ref{lem:short-basis}. This algorithm builds a datastructure in $O(n\log n)$ time that allows to query for the distance between any vertex on the boundary and any other vertex in $O(\log n)$ time. We need to query for $O(n)$ pairs of the form $(\tilde{v},\tilde{v}+\tau_\alpha+\tau_\beta)$ and retain a pair $(\tilde{v},\tilde{w})$ that minimizes the distance. In order to find an explicit representative of $\tilde{c}_\gamma$, we can in a second step run Dijkstra's algorithm with source $\tilde{v}$ in $\mathcal{D}_0\cup \mathcal{D}_1$. This yields an $O(n\log n)$ time algorithm to compute $\tilde{c}_\gamma$. Alternatively, we can glue the $p_{\alpha}$ side of $\mathcal{D}_0$ with the $\overline{p_{\alpha}}$ side of $\mathcal{D}_1$ and search for the shortest generating cycle of the resulting annulus in $O(n\log n)$ time as in~\cite[Prop. 2.7(e)]{de2010}, or, more efficiently, in time $O(n\log\log n)$ as in~\cite[Theorem 7]{italiano2011}. In the unweighted case, the shortest generating cycle corresponds to a minimum $(s,t)$-cut in a planar graph, or by duality to a maximum $(s,t)$-flow, which can be computed in linear time~\cite{w-edpup-97,ek-ltamf-13}.
      
    \item If $\tilde{v}$ lies on the side $p_{\alpha\beta}$ of $\mathcal{D}_0$, note that this side is incident to $\mathcal{D}_1$ and $\tilde{w}$ lies on the translate $p_{\alpha\beta}+\tau_\alpha+\tau_\beta$, which is also incident to $\mathcal{D}_1$. The union $p_{\alpha\beta}\cup \mathcal{D}_1\cup (p_{\alpha\beta}+\tau_\alpha+\tau_\beta)$ is convex since it is the intersection of the four supporting half-planes of $\mathcal{D}_1$. See Figure~\ref{fig:D0D1}, bottom-left. We can thus assume that $\tilde{c}_\gamma$ passes through the common origin $\tilde{v}_0$ of (lifts of) $p_{\alpha}$ and $p_{\beta}$ in $\mathcal{D}_1$. 
We can then change the origin of $\tilde{c}_\gamma$ to $\tilde{v}_0$ so that $\tilde{c}_\gamma$ is the concatenation of a shortest path between $\tilde{v}_0$ and $\tilde{w}_0$ in $\mathcal{D}_1$ with a lift of $p_{\alpha\beta}$, where  $\tilde{w}_0$ is the common target of  $\overline{p_{\alpha}}$ and $\overline{p_{\beta}}$ in $\mathcal{D}_1$. This shortest path can be computed in $O(n\log n)$ time by Dijkstra's algorithm or in linear time with Henziger et al.'s algorithm~\cite{hkrs-fspapg-97}. 
 \end{itemize}
 \item We now consider the case when $p_{\alpha\beta}$ is oriented differently in  $c_\alpha$ and $c_\beta$. We can thus write $c_\alpha = p_\alpha\cdot  \overline{p_{\alpha\beta}}$ and $c_\beta = p_\beta\cdot  p_{\alpha\beta}$ as in Figure~\ref{fig:uniCover}, right. As in the previous case, $c_\gamma$ must cross $c_\alpha$ and must have a lift $\tilde{c}_\gamma$ starting from a vertex on the side $p_\alpha$ or  $\overline{p_{\alpha\beta}}$ of a domain $\mathcal{D}_0$. This again leads to two cases as pictured in Figure~\ref{fig:D0D1-bis}.
    \begin{figure}[h]
      \centering
      \includesvg[.75\textwidth]{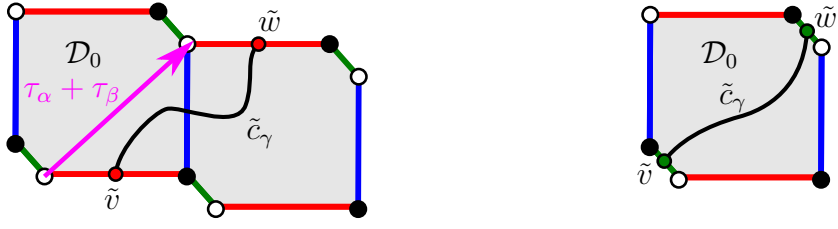}
      \caption{Left, the source vertex of $\tilde{c}_\gamma$ lies on the side $p_\alpha$ of a domain $\mathcal{D}_0$. Its target vertex thus lies on the boundary of $\mathcal{D}_1 =\mathcal{D}_0 +\tau_\alpha$. Being convex, $\mathcal{D}_0\cup \mathcal{D}_1$ contains a lift of some tight representative of $\gamma$. Right, the source vertex of $\tilde{c}_\gamma$ lies on the side $\overline{p_{\alpha\beta}}$ of  $\mathcal{D}_0$.}
      \label{fig:D0D1-bis}
    \end{figure}
    We may argue as before to compute a tight representative of $\gamma$ in $O(n\log\log n)$ time, or in linear time in the unweighted case.
  \end{itemize}
  In all cases we may compute a tight representative of $\gamma$ in $O(n\log\log n)$ time, or in linear time in the unweighted case.
\end{proof}
As in Lemma~\ref{lem:addSimple},  let $(\alpha,\beta)$ be a homology basis such that $\alpha$ and $\beta$ have representatives, respectively $c_\alpha$ and $c_\beta$, forming a good pair. We call a \textbf{line} any bi-infinite concatenation of lifts of $c_\alpha$ or of $c_\beta$ in the universal cover of $T$.
\begin{lemma}\label{lem:addSimple-bis}
  The tight representative $c_{\alpha+\beta}$ computed in Lemma~\ref{lem:addSimple} can be modified in $O(n)$ time in order to satisfy the additional following property (P):  The intersection of any lift of $c_{\alpha+\beta}$ in the universal cover with any line is either empty or a common connected subpath.
\end{lemma}
\begin{proof}
  Write $\gamma=\alpha+\beta$ as before. Since $c_\alpha$ and $c_\beta$
  are tight, every line must be geodesic.
We consider the at most five lines intersecting the convex region used to compute $\tilde{c_\gamma}$ in Lemma~\ref{lem:addSimple}. For each one in turn we compute the first and last intersections of $\tilde{c_\gamma}$ with this line and replace in $\tilde{c_\gamma}$ the subpath  in-between these intersections by the subpath between the same vertices in the line. This
does not change the length of $c_\gamma$. We still denote by $c_\gamma$ the resulting closed walk.
See Figure~\ref{fig:line-intersect}.
  \begin{figure}[h]
    \centering
    \includesvg[.65\textwidth]{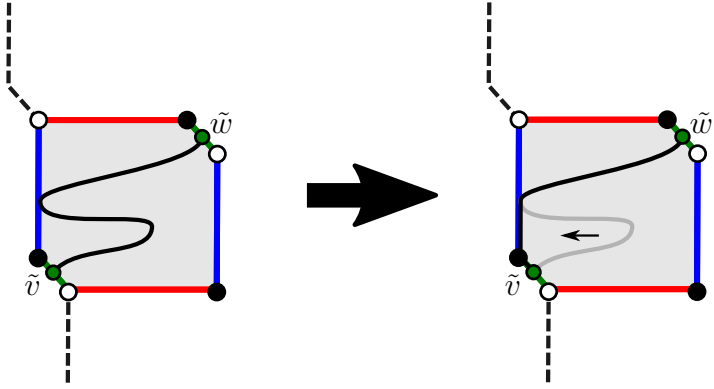}
    \caption{Removing a bigon between a line and a lift of $c_\gamma$.}
    \label{fig:line-intersect}
  \end{figure}
  When the source point $\tilde{v}$ of $\tilde{c_\gamma}$ lies on $p_\alpha$, as in Figure~\ref{fig:D0D1}, top and Figure~\ref{fig:D0D1-bis}, left, there remains the possibility that some other lift of $c_\gamma$ intersects the line separating ${\cal D}_0$ and ${\cal D}_1$ in two components. An example is given in Figure~\ref{fig:line-intersect-bis} where two consecutive lifts of $c_\gamma$ are pictured.
\begin{figure}[h]
    \centering
    \includesvg[\textwidth]{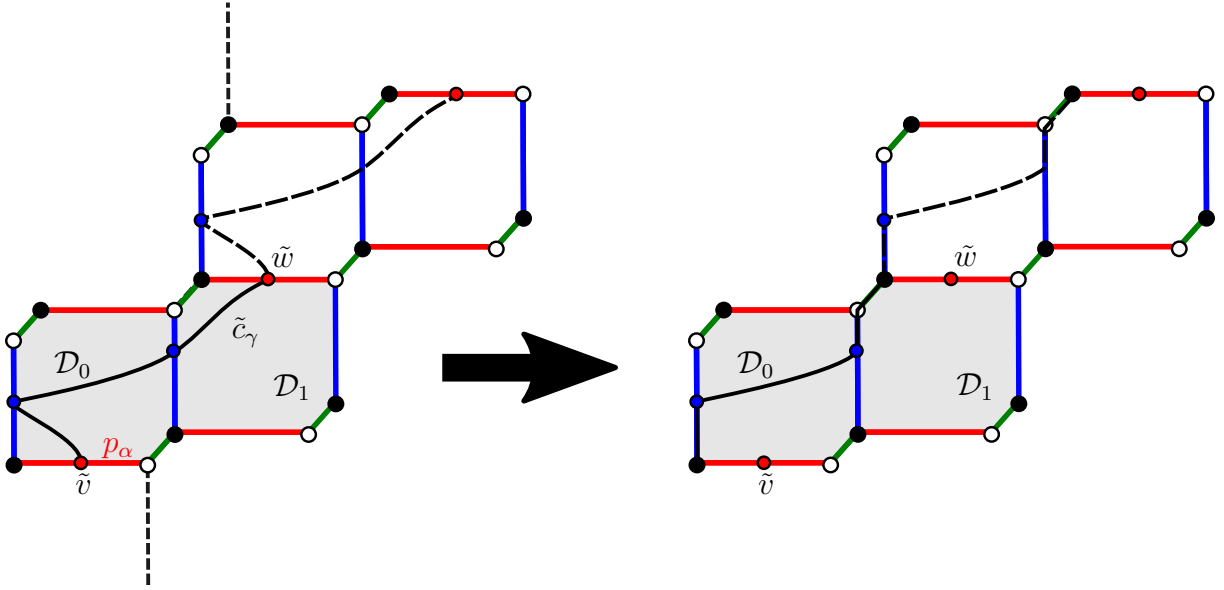}
    \caption{A bigon between a line and a lift of $c_\gamma$ that only appears when shifting the source vertex.}
    \label{fig:line-intersect-bis}
  \end{figure}
This shows that the source vertex of $c_\gamma$ could be chosen to lie on $p_{\alpha\beta}$ and we are back to the cases of Figure~\ref{fig:D0D1}, bottom and Figure~\ref{fig:D0D1-bis}, right. 

Alternatively, we can perform a last subpath substitution as in Figure~\ref{fig:line-intersect-bis} to obtain a representative $c_\gamma$ satisfying $P$. The $O(1)$ subpath substitutions take linear time.

Note that when treating the intersection of $\tilde{c_\gamma}$ with a new line $l$, we do not break the connectedness of its intersection with any of the previously treated lines $l'$. Indeed, let $d'$ be a connected path $\tilde{c_\gamma} \cap l'$ and let $d$ be a subpath of $\tilde{c_\gamma}$ between the first and the last intersection of $\tilde{c_\gamma}$ with $l$. If $l$ and $l'$ do not intersect, then $d'$ either lies inside or outside of $d$. After treating the line $l$, in the first case $\tilde{c_\gamma} \cap l'$ becomes empty, while in the second case this intersection does not change.

If $l$ and $l'$ do intersect, their intersection $l \cap l'$ is a connected path, possibly reduced to a vertex. There are several cases to consider, according to the relative position of the paths $d'$ and $d$:
\begin{itemize}
    \item Figure \ref{fig:simplifying_cases}, top left: if $d'$ lies outside of $d$, then $\tilde{c_\gamma} \cap l'$ is not affected by treating the line $l$;
    \item Figure \ref{fig:simplifying_cases}, top right: if $d'$ contains $d$, then the endpoints of $d$ belong to $l \cap l'$, and so $d$ itself is a subpath of $l \cap l'$ and there is no change to be done;
    \item if $d'$ is contained in $d$, there are two subcases:
        \begin{itemize}
            \item Figure \ref{fig:simplifying_cases}, bottom left: if the endpoints of $d$ lie on one side from $l \cap l'$ on the line $l$, then the intersection $\tilde{c_\gamma} \cap l'$ becomes empty;
            \item Figure \ref{fig:simplifying_cases}, bottom right: if the endpoints of $d$ lie on different sides of $l \cap l'$, then $\tilde{c_\gamma} \cap l'$ becomes $l \cap l'$;
        \end{itemize}
    \item Figure \ref{fig:simplifying_cases}, far right: if one endpoint of $d'$ is contained in $d$, while the other is not, then the endpoint inside belongs to $l \cap l'$, and so $\tilde{c_\gamma} \cap l'$ will become a concatenation of $d' \setminus d$ and a subpath of $l \cap l'$.
\end{itemize}

\begin{figure}[h]
    \centering
    \includesvg[\textwidth]{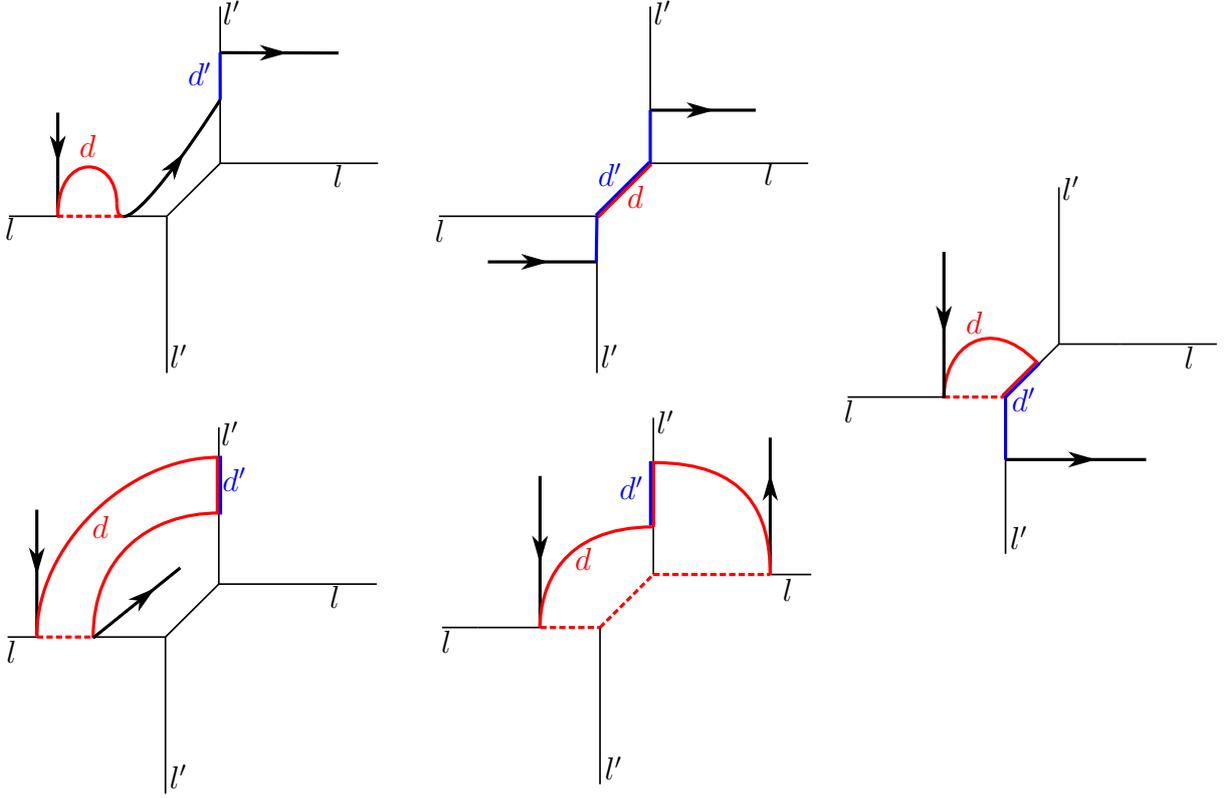}
    \caption{The lines $l$ and $l'$ intersect along a connected path. The path $\tilde{c_\gamma}$ is in bold. After treating the line $l$, $d$ is replaced by the red dashed path.}
    \label{fig:simplifying_cases}
  \end{figure}
\end{proof}

\begin{lemma}\label{lem:good}
  Let $(\alpha,\beta)$ be a homology basis such that $\alpha$ and $\beta$ have representatives, respectively $c_\alpha$ and $c_\beta$, forming a good pair. If $\Norm{G,w}(\alpha+\beta) < \Norm{G,w}(\alpha) + \Norm{G,w}(\beta)$, then the tight representative $c_{\alpha+\beta}$ computed in Lemma~\ref{lem:addSimple-bis} is such that $(c_\alpha,c_{\alpha+\beta})$ and $(c_{\alpha+\beta},c_\beta)$ are good pairs.
\end{lemma}
\begin{proof}
  We first show that $c_\gamma$ is simple, where as usual $\gamma=\alpha+\beta$. Suppose otherwise by way of a contradiction. Then, the lift $\tilde{c_\gamma}$ of $c_\gamma$ considered in Lemma~\ref{lem:addSimple-bis} must contain two vertices $\tilde{u}$ and $\tilde{u}'$ projecting onto the same vertex $u$ in $G$ and such that $\{\tilde{u}, \tilde{u}'\}\neq \{\tilde{v}, \tilde{w}\}$, where $\tilde{v}$ and $\tilde{w}$ are the source and target vertex of $\tilde{c_\gamma}$, respectively. In particular, $\tilde{u}' = \tilde{u}+\tau$ for some covering translation $\tau$. We thus have the decomposition $\tilde{c_\gamma} = p_1\cdot p_2\cdot p_3$, where $p_1$ is a path from $\tilde{v}$ to $\tilde{u}$, $p_2$ is a path from $\tilde{u}$ to $\tilde{u}'$, $p_3$ is a path from $\tilde{u}'$ to $\tilde{w}$, and at most one of these paths is reduced to a vertex. From the proof of Lemma~\ref{lem:addSimple}, there are four cases to consider corresponding to Figure~\ref{fig:D0D1}, top and bottom, and Figure~\ref{fig:D0D1-bis}, left and right.
  \begin{itemize}
  \item 
    In Figure~\ref{fig:D0D1}, bottom, and Figure~\ref{fig:D0D1-bis}, right, we may only have $\tau=\tau_\lambda$ for\\ $\lambda\in\{0,\pm\alpha,\pm\beta,\pm\gamma,\pm(\alpha-\beta)\}$ as otherwise the region containing  $\tilde{c_\gamma}$ is disjoint from its translate by $\tau$. We show below that each of these values for $\tau$ leads to a contradiction.
    \begin{itemize}
    \item We cannot have $\tau=0$, as this would imply $\tilde{u}=\tilde{u}'$ in contradiction with the simplicity of the shortest path $\tilde{c_\gamma}$.
    \item If $\tau=\pm\tau_\gamma$, then $p_2$ or $\overline{p_2}$ is homologous to $\gamma$ but shorter than $\tilde{c_\gamma}$, in contradiction with the assumption that $c_\gamma$ is tight.
      \item If $\tau=\tau_\alpha$ then,
  putting $p=p_1\cdot (p_3-\tau)$, we have that $q(p_2)$ is homologous to $c_\alpha$ and $q(p)$  is homologous to $c_\beta$. (Recall that $q:\mR^2\to T$ is the universal cover projection on the torus.) We infer $w(p_2)\geq \Norm{G,w}(\alpha)$ and $w(p)\geq \Norm{G,w}(\beta)$. It follows then from tightness of $c_\gamma$ that
  \[\Norm{G,w}(\gamma) = w(c_\gamma) = w(p_2)+w(p) \geq \Norm{G,w}(\alpha) + \Norm{G,w}(\beta).
  \]
  This however contradicts the hypothesis in the lemma.
\item The case $\tau=\tau_\beta$ can be excluded analogously.
\item If  $\tau=-\tau_\alpha$ then we must have $\tilde{u}$ and $\tilde{u}'$ respectively on the $\overline{p_\beta}$ side and the $p_\beta$ side of 
  the domain containing $\tilde{c_\gamma}$. See Figure~\ref{fig:D0D1-ter}.
\begin{figure}[h]
      \centering
      \includesvg[.75\textwidth]{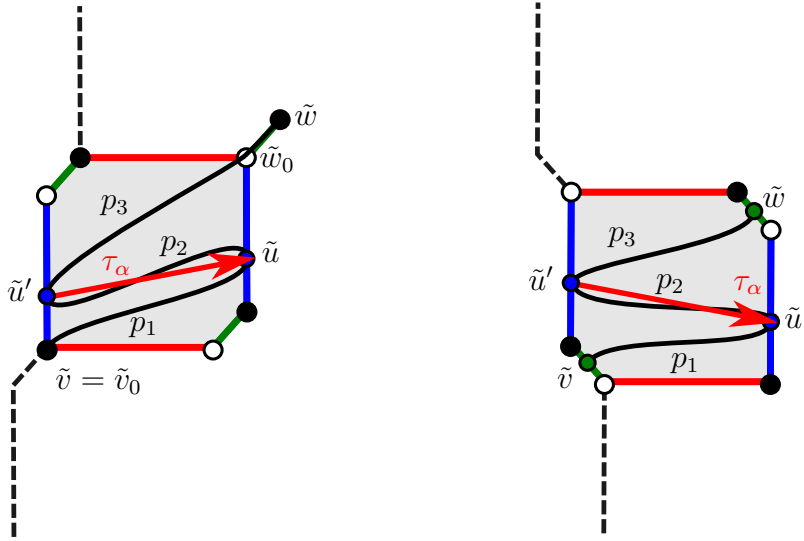}
      \caption{The case $\tau=-\tau_\alpha$. Left: $p_{\alpha\beta}$ has the same orientation in $c_\alpha$ and $c_\beta$. Right: $p_{\alpha\beta}$ has opposite orientations in $c_\alpha$ and $c_\beta$.}
      \label{fig:D0D1-ter}
    \end{figure}
  It follows that $p_1\cdot p_2$ has its endpoints $\tilde{v}$ and $\tilde{u}'$ on the same line, while the intermediate vertex $\tilde{u}$ is not on this line. This contradicts Property (P) of Lemma~\ref{lem:addSimple-bis}.
  The case $\tau=-\tau_\beta$ similarly leads to a contradiction.
\item Finally, if $\tau=\tau_{\alpha-\beta}$ or $\tau=\tau_{\beta-\alpha}$, only the case of bottom Figure~\ref{fig:D0D1} can occur since in the case of Figure~\ref{fig:D0D1-bis}, right, ${\cal D}_0$ is disjoint from ${\cal D}_0+\tau$. In the former case, $\tilde{u}$ and $\tilde{u}'$ must lie on the sides $p_{\alpha\beta}$ or $\overline{p_{\alpha\beta}}$ of the domain. After shifting the source point $\tilde{v}$ to a translate of $\tilde{u}$, we have that either $[q(p_1)]=\alpha$ and $[q(p_2\cdot p_3)]=\beta$, or $[q(p_1)]=\beta$ and $[q(p_2\cdot p_3)]=\alpha$. In both case we may argue as for $\tau=\tau_\alpha$ that this contradicts the hypothesis $\Norm{G,w}(\alpha+\beta) < \Norm{G,w}(\alpha) + \Norm{G,w}(\beta)$.
    \end{itemize}
  \item In Figure~\ref{fig:D0D1}, top, and Figure~\ref{fig:D0D1-bis}, left, we may only have $\tau = \tau_\lambda$ for \\
    $\lambda \in \{0, \pm\alpha, \pm\beta,\pm\gamma,\pm (\alpha-\beta),\pm{2\alpha}, \pm (2\alpha-\beta), \pm (2\alpha+\beta)\}$ as otherwise  the region ${\cal D}_0\cup{\cal D}_1$ containing  $\tilde{c_\gamma}$ is disjoint from its translate by $\tau$. We again show that each possible value for $\lambda$ leads to a contradiction.
    \begin{itemize}
    \item The cases $\tau = \tau_\lambda$ for $\lambda \in \{0, \pm\alpha, \pm\beta,\pm\gamma\}$ can be excluded with the same arguments as above. We can also exclude $\tau = \tau_{-2\alpha}$ with a similar argument as for $\tau = \tau_{-\alpha}$.
\item In the case $\tau = \tau_{2\alpha}$, a simple application of Jordan's curve theorem implies that $p_2$ and $p_2+\tau_\alpha$ must intersect, say at a vertex $\tilde{x}'$. Then $\tilde{x}:= \tilde{x}'-\tau_\alpha$ belongs to $(p_2+\tau_\alpha)-\tau_\alpha = p_2$. We are thus reduced to the cases $\tau = \tau_\alpha$ or $\tau = -\tau_\alpha$ (depending on which of $\tilde{x}, \tilde{x}'$ occurs first along
    $\tilde{c_\gamma}$) by substituting $\tilde{u}, \tilde{u}'$ with $\tilde{x}, \tilde{x}'$. See  Figure~\ref{fig:D0D1-quater}.
      \begin{figure}[h]
      \centering
      \includesvg[\textwidth]{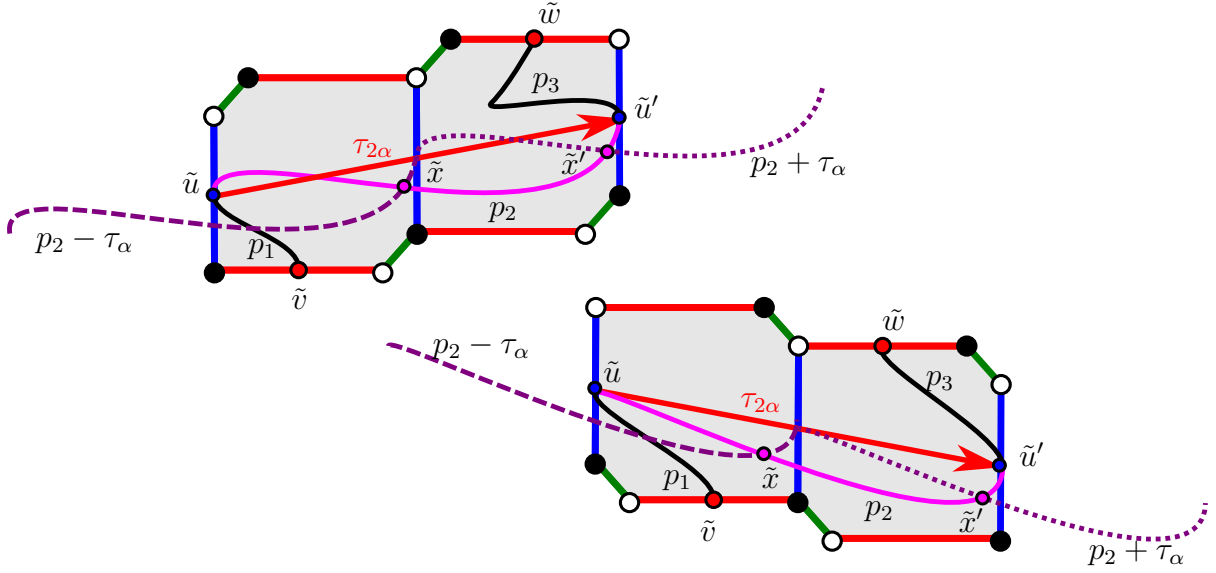}
     \caption{The case $\tau=\tau_{2\alpha}$. Left: $p_{\alpha\beta}$ has the same orientation in $c_\alpha$ and $c_\beta$. Right: $p_{\alpha\beta}$ has opposite orientations in $c_\alpha$ and $c_\beta$.}
      \label{fig:D0D1-quater}
    \end{figure}
  \item If $\tau=\tau_{\alpha-\beta}$, then $\tilde{u}'$ must be on the same $a$-line as $\tilde{v}$, while $\tilde{u}$ is on the other $a$-line through $\tilde{w}$. It follows that $p_1\cdot p_2$ has its endpoints on a same line but an intermediate vertex on another line, in contradiction with property $P$ of Lemma~\ref{lem:addSimple-bis}. See Figure~\ref{fig:D0D1-quinquies} for the case of  top Figure~\ref{fig:D0D1} where $p_{\alpha\beta}$ is oriented the same way in $c_\alpha$ and $c_\beta$.
    \begin{figure}[h]
      \centering
       \includesvg[\textwidth]{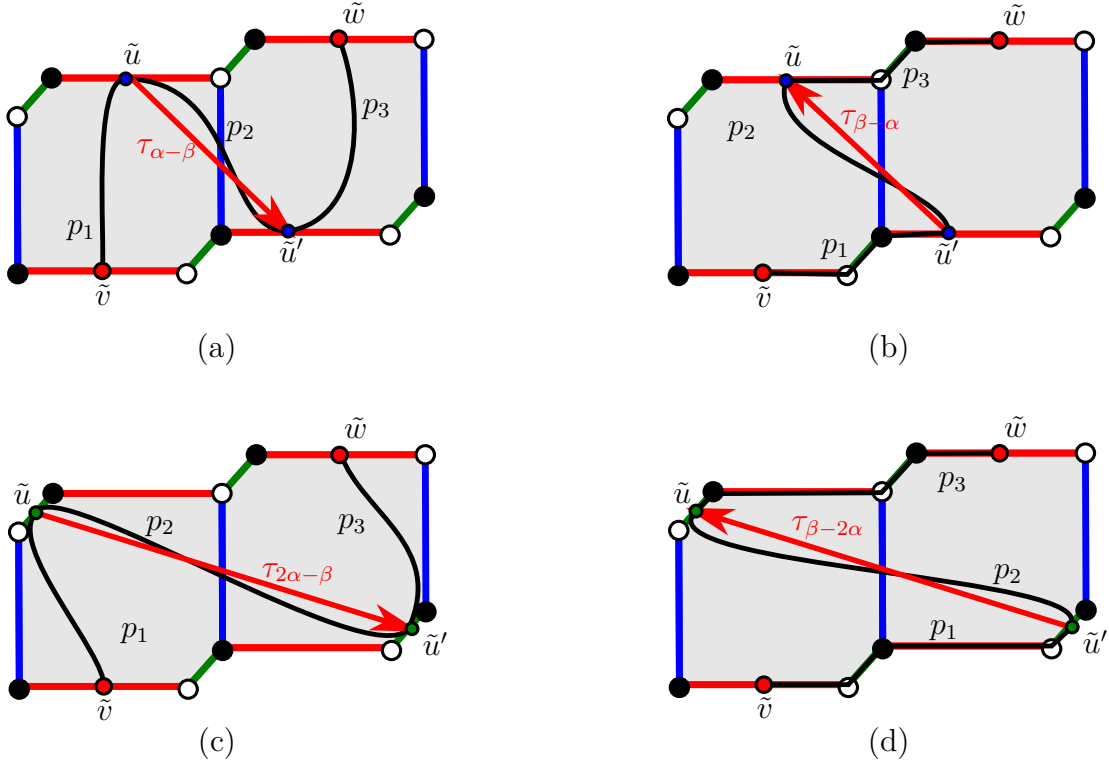}
      \caption{The case of top Figure~\ref{fig:D0D1} with  $\tau=\tau_{\alpha-\beta}$ (a), $\tau=\tau_{\alpha-\beta}$ (b), $\tau=\tau_{2\alpha-\beta}$ (c), and $\tau=\tau_{\beta-2\alpha}$ (d).}
      \label{fig:D0D1-quinquies}
    \end{figure}
  \item If $\tau=\tau_{\beta-\alpha}$, then $\tilde{u}$ must be on the same $a$-line as $\tilde{v}$, and $\tilde{u}$ is on the same $a$-line as $\tilde{w}$. Property (P) forces $p_1$ (and $p_3$) to intersect the side $p_{\alpha\beta}$. We are thus back to the cases of bottom Figure~\ref{fig:D0D1} and right Figure~\ref{fig:D0D1-bis} previously discussed.
  \item if $\tau=\pm\tau_{2\alpha-\beta}$, only the case of bottom Figure~\ref{fig:D0D1} can occur since in the case of Figure~\ref{fig:D0D1-bis}, right, ${\cal D}_0\cup {\cal D}_1$ is disjoint from ${\cal D}_0\cup {\cal D}_1+\tau$. Then, for both values of $\tau$, (the lift of) $\tilde{c}_\gamma$ must go through $p_{\alpha\beta}$ and we are  back to the case of bottom Figure~\ref{fig:D0D1}.
    \item Finally, $\tau=\pm\tau_{2\alpha+\beta}$ can only occur for the case of Figure~\ref{fig:D0D1-bis}, left, and similarly to the previous case, $\tilde{c}_\gamma$ must go through $p_{\alpha\beta}$ and we are  back to the case of right Figure~\ref{fig:D0D1-bis}.
  \end{itemize}
    \end{itemize}
    We have thus proved that $c_\gamma$ is a simple cycle.
      It remains to check that it intersects each of $c_\alpha$ and $c_\beta$ along a connected path. The convex region made of one or two domains that contains $\tilde{c_\gamma}$ intersects exactly two horizontal lines in each of the cases pictured in Figures~\ref{fig:D0D1} and~\ref{fig:D0D1-bis}. By property (P), the intersections of $\tilde{c_\gamma}$ with these two lines are each connected and their projections share the same vertex $v$, i.e., the source vertex of  $c_\gamma$. It follows that $c_\gamma\cap c_\alpha$ is connected.
  For the cases of Figure~\ref{fig:D0D1}, bottom, and Figure~\ref{fig:D0D1-bis}, right,
  we can similarly argue that $c_\gamma\cap c_\beta$ is connected. In the two remaining cases, as in Figure~\ref{fig:D0D1}, top, and Figure~\ref{fig:D0D1-bis}, left, if $\tilde{c_\gamma}$ would intersect two vertical lines, then considering two successive lifts of $\tilde{c_\gamma}$ as in Figure~\ref{fig:line-intersect-bis} we would get a non-connected intersection with the middle vertical line. This would however contradict Property (P).
\end{proof}

Since in Algorithm~\ref{alg:algo1} we only add $\alpha+\beta$ to $H$ at line~\ref{alg:addH} when $\alpha,\beta$ are already in $H$ with $\Norm{G,w}(\alpha+\beta) < \Norm{G,w}(\alpha) + \Norm{G,w}(\beta)$, Lemma~\ref{lem:good} immediately implies
\begin{corollary}\label{cor:H}
  $H\subset \TSimpleC\subset \SimpleC$
\end{corollary}

\begin{corollary}\label{cor:n-iterations}
  The number of iterations in the while loop of Algorithm~\ref{alg:algo1}, from Line~\ref{alg:while} to~\ref{alg:endwhile}, is bounded by twice the size of $\TSimpleC$.
\end{corollary}
\begin{proof}
  The number $N$ of iterations in the while loop is bounded by the number of times a sector can be removed from $S$ in Line~\ref{alg:removeS}. This number is at most two, the initial size of $S$, plus the number of sectors inserted in $S$ at Line~\ref{alg:addS}. This last number is itself twice the number of triples inserted in $H$ at line~\ref{alg:addH}, not taking into account the two classes initially put in $H$ at line~\ref{alg:initH}. Note that the inserted triples are pairwise distinct. It ensues from Corollary~\ref{cor:H} that $N\leq 2 + 2(|H|-2)\leq 2|\TSimpleC|$.
\end{proof}

\begin{proposition}\label{prop:algo}
  Algorithm~\ref{alg:algo1} runs in $O(n^2\log\log n)$ time. This reduces to $O(n^2)$ in the case of unweighted graphs.
\end{proposition}
\begin{proof}
  From Corollary~\ref{cor:n-iterations}, Algorithm~\ref{alg:algo1} enters at most $2|\TSimpleC|$ times the while loop between Lines~\ref{alg:while} and~\ref{alg:endwhile}. This is $O(n)$ iterations by Lemma~\ref{lem:numberofsimplecycles}. Each iteration takes $O(n\log\log n)$ time for executing Line~\ref{alg:tight} by Lemmas~\ref{lem:addSimple} and~\ref{lem:addSimple-bis}, or linear time in the case of unweighted graphs. Lemma~\ref{lem:good} ensures that only good pairs are stored at Line~\ref{alg:addS}, so that the requirement for executing Line~\ref{alg:tight} is always satisfied. Line~\ref{alg:symmetrize} moreover takes time $O(|H|)=O(n)$. Since every other line takes constant time to execute, the total time for running Algorithm~\ref{alg:algo1} is $O(n\cdot n\log\log n + n)=O(n^2\log\log n)$. In the unweighted case this reduces to $O(n^2)$ time.
\end{proof}
We are now ready to prove that the unit ball $\Ball{G,w}$ can be computed in $O(n^2\log\log n)$ time. As before, the $\log\log n$ factor can be omitted in the unweighted case. Henceforth, we will not mention this again.
\begin{proof}[Proof of Theorem~\ref{th:unit-ball}]
  Proposition~\ref{prop:algo} states that we can compute in $O(n^2\log\log n)$ time a list $H$ of $O(n)$ vectors, with their norms, that contains the directions of the extremal points of the unit ball $\Ball{G,w}$. After normalising the vectors we compute their convex hull in $O(n\log n)$ time with any classical convex hull algorithm. (In fact, we can easily maintain an angular ordering of the elements of $H$, so that the convex hull of its normalized vectors can be computed in linear time.) Overall, this leads to an $O(n^2\log\log n)$ time algorithm for computing $\Ball{G,w}$.
\end{proof}
For further reference, we establish a useful property of the ordered set of elements in $H$. Namely, two consecutive cycles in $H$ define a \define{unimodular} cone.
\begin{lemma}\label{lem:unimodularPartition}
  The sorted list $H= [((x_i,y_i), c_i, w(c_i))]$ computed by Algorithm~\ref{alg:algo1} is such that the rays $\mR_{\geq 0} c_i$ are ordered cyclically by angle, and $\langle [c_i], [c_{i+1}] \rangle = 1$ for all $i$. In particular, the half-open cones $C_i = \mZ_{\ge 0} [c_i] + \mZ_{> 0} [c_{i+1}]$ constitute a partition of $H_1(T; \mZ)$. 
\end{lemma}
\begin{proof}
  That the rays are ordered by angle follows by induction and from the insertion of $h''$ in-between $h$ and $h'$ in Line~\ref{alg:addH}. Indeed the ray corresponding to $h''$ falls between the rays corresponding to $h$ and $h'$. Moreover, Lemma~\ref{lem:good} ensures that the corresponding pairs of homology classes $([c],[c''])$ and $([c''],[c'])$ are good, thus forming homology bases. Of course, all the lattice points in the corresponding cones have non-negative coordinates.
\end{proof}

\section{Computing the length spectrum}\label{sec:lengthSpectrumComputation}

We now give a proof of Theorem~\ref{thm:lengthSpectrumComputation}.

\begin{proof}
	First, as described in Sections~\ref{sec:short-basis} and~\ref{sec:computingUnitBall}, we compute in time $O(n^2\log\log n)$ a short basis $(a,b)$ and a list $H$ of triples $((x_i,y_i),c_i,w(c_i))$, where $c_i$ is a simple tight cycle in $G$, $(x_i,y_i)\in\mZ^2$ represents its homology class $[c_i]=x_i[a]+y_i[b]$ and $w(c_i)=\Norm{G,w}([c_i])$. From Lemma~\ref{lem:unimodularPartition}, this provides a partition of $H_1(T; \mZ)$ into half-open cones $C_i = \mZ_{\ge 0} [c_i] + \mZ_{> 0} [c_{i+1}]$. Note that by construction $\Norm{G,w}$ is linear over each $C_i$.
	
    We then compute the values of the length spectrum iteratively, in increasing order, storing them into a list $\Lambda$, initially empty.  Intuitively, the algorithm consists in sweeping $\HoneT$ by increasing the radius of the  $\lambda$-ball $\lambda\Ball{G,w}$ from $\lambda=0$. Each time a lattice point is swept, its norm $\lambda$ is added to $\Lambda$. We actually sweep the cones $C_i$ in parallel. 
        
    By Lemma~\ref{lem:unimodularPartition},  the half-open cones $C_i$ decompose  the ball $\lambda\Ball{G,w}$  into sectors $C_i^\lambda:=\lambda\Ball{G,w}\cap C_i$. For each sweeping value $\lambda$ and each $i$ we store two ordered subsets of $C_i^\lambda$ into dequeues (double-ended queues) $F^h_i$ and $F^v_i$ corresponding to the \emph{horizontal} and \emph{vertical} sweeping front, respectively. Formally $F^h_i = ((x_1^h,y_1^h),\dots, (x_{i_h}^h,y_{i_h}^h))$, where each $(x_j^h,y_j^h)\in \mZ_{\ge 0}\times\mZ_{> 0}$ is such that $x_j^h[c_i]+y_j^h[c_{i+1}]\in C_i^\lambda$ and $(x_j^h+1)[c_i]+y_j^h[c_{i+1}]\not\in \lambda\Ball{G,w}$. Moreover,  the homology classes in $F^h_i$  are ordered by their norms in increasing order. Similarly, $F^v_i = ((x_1^v,y_1^v),\dots, (x_{i_v}^v,y_{i_v}^v))$ contains the list (ordered by norm in increasing order) of coordinates of the homology classes contained in $\lambda\Ball{G,w}$ but whose translates by $[c_{i+1}]$ have norms larger than $\lambda$. Initially $F^h_i$ is the empty dequeue and $F^v_i$ contains the coordinates $(0,0)$ of the zero class. See Figure~\ref{fig:cone-sweep}
    \begin{figure}[h]
      \centering
       \includegraphics[width=.6\linewidth]{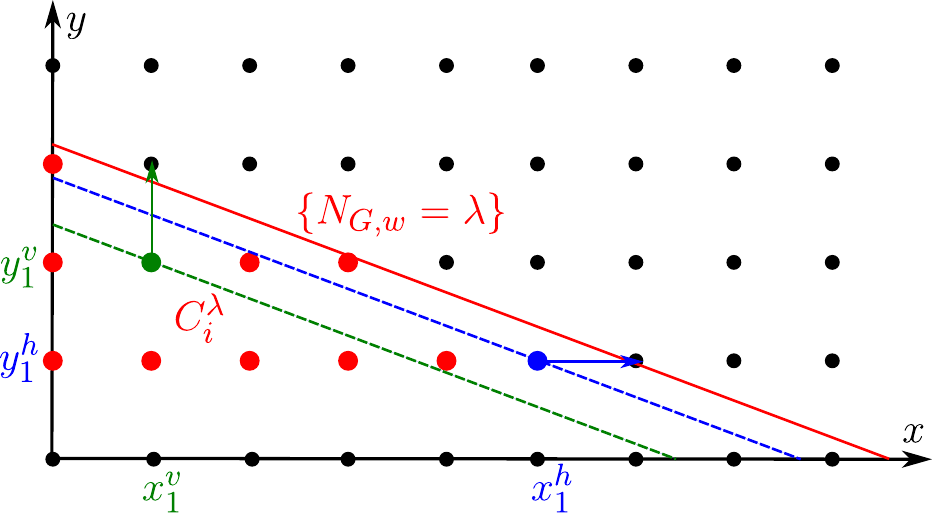}
      \caption{In the figure, we represent the class $x[c_i]+y[c_{i+1}]$ by the point with coordinates $(x,y)$. The solid red line represents the points whose norm is $\lambda$. The dashed lines correspond to the norms of the points $(x_1^h,y_1^h)$ and $(x_1^v,y_1^v)$. }
      \label{fig:cone-sweep}
    \end{figure}
        \begin{claim}
          The coordinates in the $([c_i],[c_{i+1}])$ basis of the homology class in $C_i\setminus C_i^\lambda$ with the smallest norm is either $(x_1^h+1,y_1^h)$ or $(x_1^v,y_1^v+1)$.
        \end{claim}
        \begin{proof}[Proof of the claim]
          Let $\alpha$ be the homology class of $C_i\setminus C_i^\lambda$ with the smallest norm. If $\alpha-[c_{i+1}]\in C_i$, then by linearity of the norm in $C_i$, we have
          \[ \Norm{G,w}(\alpha-[c_{i+1}])=\Norm{G,w}(\alpha)-\Norm{G,w}([c_{i+1}])<\Norm{G,w}(\alpha).
          \]
	It follows by minimality of the norm of $\alpha$ that $\alpha-[c_{i+1}]\in C_i^\lambda$. If $\alpha-[c_{i+1}]\notin C_i$, then we must have $\alpha\in \mZ_{>0}[c_i]+[c_{i+1}]$. It follows that $\alpha-[c_i]\in C_i$ and
	\[ \Norm{G,w}(\alpha-[c_{i}])=\Norm{G,w}(\alpha)-\Norm{G,w}(([c_{i}])<\Norm{G,w}(\alpha),\]
	implying $\alpha-[c_{i}]\in C_i^\lambda$. Hence, $\alpha$ is of the form $\beta+[c_i]$ or $\beta+[c_{i+1}]$ for some $\beta\in C_i^\lambda$. In other words the coordinates of $\alpha$ have either the form $(x+1,y)$ for some $(x,y)\in F^h_i$ or the form $(x,y+1)$ for some $(x,y)\in F^v_i$. Of course, since $\alpha$ has minimal norm in $C_i\setminus C_i^\lambda$, it must be that $(x,y)=(x_1^h,y_1^h)$ in the former case and $(x,y)=(x_1^v,y_1^v)$ in the latter case. 
      \end{proof}
      We are now ready to describe the sweeping algorithm. We store the indices $i$ of the sectors $C_i^\lambda$ in a (balanced) binary search tree $\cal S$ allowing minimum extraction, deletion and insertion in logarithmic time. The key of sector $C_i^\lambda$ used for comparisons in the tree is the minimum norm of a homology class in $C_i\setminus C_i^\lambda$. From the previous claim it can be computed in constant time from $F^h_i$ and $F^v_i$ as
      \begin{align*}
        \text{key}[i]&=\min(\Norm{G,w}([(x_1^h+1)c_i+y_1^hc_{i+1}]), \Norm{G,w}([x_1^vc_i+(y_1^v+1)c_{i+1}]))\\
        &=\min((x_1^h+1)w(c_i)+y_1^hw(c_{i+1}), x_1^vw(c_i)+(y_1^v+1)w(c_{i+1}))
      \end{align*}
      Suppose we have computed the $m$ first values of the length spectrum, i.e., $\Lambda = (\lambda_1,\lambda_2,\dots,\lambda_m)$, and we want to compute $\lambda_{m+1}$. We extract and remove from $\cal S$ the sector $C_i^{\lambda_m}$, with $i=\cal S.\min()$, i.e., with a non-swept homology class $\alpha$ of minimal norm. Hence, we  have $\lambda_{m+1}=\text{key}[i]$.
      We update $F^h_i$ and $F^v_i$ as follows. If $\alpha=(x_1^h+1)[c_i]+y_1^h[c_{i+1}]$, then we remove $(x_1^h,y_1^h)$ from the bottom of $F^h_i$ and push $(x_1^h+1,y_1^h)$ on its top. Likewise, if $\alpha=x_1^v[c_i]+(y_1^v+1)[c_{i+1}]$, we remove $(x_1^v,y_1^v)$ from the bottom of $F^v_i$ and push $(x_1^v,y_1^v+1)$ on its top. We do both if $\alpha=(x_1^h+1)[c_i]+y_1^h[c_{i+1}]=x_1^v[c_i]+(y_1^v+1)[c_{i+1}]$. Clearly, the updated dequeues contain the required lattice points with respect to the sweeping value $\lambda=\lambda_{m+1}$.
      We then update $\cal S$ by inserting $i$ with its new key resulting from the updates of $F^h_i$ and $F^v_i$. We finally add $\lambda_{m+1}$ to $\Lambda$. Note that some values of the length spectrum might have multiplicity larger than one. In this case, $\lambda$ stays constant for some iterations while the sweeping fronts are updated (until we reach a larger value of the length spectrum). By Corollary~\ref{cor:H} and Theorem~\ref{thm:weightedgraphspolyhedralnorm}, $\cal S$ contains $O(n)$ items. The running time for a sweeping step is thus $O(\log n)$ time for interacting with $\cal S$ plus constant time for updating $F^h_i$, $F^v_i$ and $\Lambda$. We can thus compute the first $k$ values of the length spectrum in $O(n^2\log\log n + k\log n)$ time.
      \end{proof}

\begin{remark}
    Note that the algorithm described in the proof of Theorem~\ref{thm:lengthSpectrumComputation} can not only compute the first $k$ values of the length spectrum, but also their homology classes at no extra cost.
  \end{remark}

\section{Deciding equality of length spectra}\label{sec:lengthSpectraEquality}

We now present an application of Algorithm~\ref{alg:algo1} to the following decision problem: given two weighted graphs $(G,w)$ and $(G',w')$ embedded on tori do they have the same length spectra? This question actually covers two problems: the equality of the marked length spectrum and the equality of the unmarked length spectrum. As we show, the former reduces to the linear equivalence of polyhedral norms which has a straightforward quadratic time solution. In contrast the latter reduces to the problem of polynomial identity testing which is only known to be in the co-RP complexity class. In particular, this problem is in co-NP. For unweighted graphs however, the complexity is deterministic polynomial. 

We first aim to compare the length spectrum as maps from $H_1(T; \mZ) \to \mR$. We are given two weighted graphs $(G,w)$ and $(G',w')$ embedded on tori $T$ and $T'$ respectively. We say that $(G,w)$ and $(G',w')$ have the same \emph{marked spectrum} if there exists a homeomorphism $\phi: T \to T'$ such that for all $\gamma \in H_1(T; \mZ)$ we have $\Norm{G',w'}(\phi_*(\gamma)) = \Norm{G,w}$ where $\phi_*(\gamma)$ denotes the class $[\phi(c)]$ where $c$ is a curve representative of $\gamma$.
\begin{theorem} \label{thm:decidingEqualityMarkedSpectrum}
Let $(G,w)$ and $(G',w')$ be two weighted graphs cellularly embedded on tori $T$ and $T'$, each with complexity bounded by $n$. Then there is an algorithm that answers whether $(G,w)$ and $(G',w')$ have the same marked spectrum in time $O(n^2\log\log n)$.
\end{theorem}

\begin{proof}
By Theorem~\ref{th:unit-ball}, one can compute the unit balls $\Ball{G,w}$ and $\Ball{G',w'}$ in time $O(n^2 \log\log n)$. A homeomorphism $\phi: T \to T'$ induces an invertible linear isomorphism $\phi_*: H_1(T; \mZ) \to H_1(T'; \mZ)$ and conversely, given a linear isomorphism $H_1(T; \mZ) \to H_1(T'; \mZ)$ there exists a homeomorphism that induces it. Hence, after identifying $\Ball{G,w}$ and $\Ball{G',w'}$ to subsets of $\mR^2$ by using the short bases $(a,b)$ and $(a',b')$ we are left with the following problem: finding a matrix $M \in \operatorname{GL}_2(\mZ)$ such that $M\cdot \Ball{G,w} = \Ball{G',w'}$.

Let us first note that linear maps preserve extremal points of convex sets. In particular, if the marked spectra are identical the number of extremal points of $\Ball{G,w}$ and $\Ball{G',w'}$ must be the same. We can hence assume that $\Ball{G,w}$ and $\Ball{G',w'}$ have $t = O(n)$ vertices by Theorem~\ref{thm:weightedgraphspolyhedralnorm}.

Now, fix two consecutive vertices $v_0, v_1$ of $\Ball{G,w}$. For each pair of consecutive vertices (in either clockwise or counterclockwise order) $v'_0, v'_1$ of $\Ball{G',w'}$ compute the matrix $M$ such that $M v_0 = v'_0$ and $M v_1 = v'_1$. If this matrix has determinant $\pm 1$ and $M\cdot \Ball{G,w} = \Ball{G',w'}$ (which can be checked by mapping the $t$ vertices) then we found a linear equivalence. If otherwise none of the pairs $(v'_0, v'_1)$ worked the balls are not linearly equivalent and the marked spectra are different. Since the algorithm required the analysis of $2t$ pairs, the overall additional cost is $O(t^2) = O(n^2)$.
\end{proof}

Now we consider the more delicate question of comparing unmarked length spectra. That is, we want to decide whether the list of values $\{\Norm{G,w}(\alpha) : \alpha \in H_1(T; \mZ)\}$ and $\{\Norm{G',w'}(\alpha') : \alpha' \in H_1(T'; \mZ)\}$ coincide where each value comes with multiplicity according to the number of homology classes that realize this length. This equality of unmarked length spectra is always decidable and we show that it belongs to the co-RP complexity class, i.e. our algorithm can detect if the unmarked spectra of $(G,w)$ and $(G',w')$ are different in random polynomial time.

For this specific test, we need to have access to all integral linear relations between the weights at once. That is, our algorithm needs to have access to the $\mQ$-vector space $\{(x_e)_{e \in E(G)} \in \mQ^{E(G)}: \sum_e x_e w_e = 0\}$. We assume that the weights are given in the following form : we are given $r$ real numbers $o = (o_1, o_2, \ldots, o_r)$ that do not satisfy any integral linear relations, and for each edge $e \in E(G)$ its weight is given as a linear combination of these real numbers with integral coefficients $w_e = w_{e,1} o_1 + \ldots + w_{e,r} o_r$. We call \emph{complexity} of these weights the sum $\|w\|_o := \sum_{e \in E(G)} \sum_{i=1}^r |w_{e,i}|$. Note that this complexity depends on the choice of the numbers $o_1, \ldots, o_r$ and not only on the values $w_e$ as real numbers.

\begin{theorem} \label{thm:decidingEqualityUnmarkedSpectrum}
Let $(G,w)$ and $(G',w')$ be two weighted graphs cellularly embedded on tori $T$ and $T'$, each with complexity bounded by $n$, where each weight is specified as $w_e = w_{e,1} o_1 + \ldots + w_{e,r} o_r$ with $w_{e,i} \in \mZ$ and $o_1, \ldots, o_r$ are $r$ given real numbers.
There is an algorithm to decide whether $(G,w)$ and $(G',w')$ have different (unmarked) spectra that runs in random polynomial time in the total input size $n + \log\left( \|w\|_o + \|w'\|_o \right)$. Moreover,
for fixed $r$, there is a deterministic algorithm that runs in time $O(n^2 \cdot (\|w\|_o + \|w'\|_o)^r)$.
\end{theorem}

Let us first explain how to deduce Theorem~\ref{thm:decidingEqualitySpectrumUnweighted} from Theorem~\ref{thm:decidingEqualityMarkedSpectrum} and Theorem~\ref{thm:decidingEqualityUnmarkedSpectrum}. Let us emphasize that Theorem~\ref{thm:decidingEqualityUnmarkedSpectrum} allows us to deduce deterministic polynomial time only in the unweighted case. Even with rational weights we are not aware of a deterministic polynomial time algorithm.
\begin{proof}[Proof of Theorem~\ref{thm:decidingEqualitySpectrumUnweighted}]
Let $G$ and $G'$ be unweighted graph with total size $n$.

The case of marked spectrum is simply a particular case of Theorem~\ref{thm:decidingEqualityMarkedSpectrum} where we use the fact that the unit ball can be computed in quadratic time in the unweighted case.

For the equality of unmarked spectrum, we have $r=1$ and $o_1=1$. The second part of Theorem~\ref{thm:decidingEqualityUnmarkedSpectrum} hence gives deterministic polynomial time in $O(n^2 \cdot (\sum_{e \in E(G) \cup E(G')} 1)) = O(n^3)$.
\end{proof}

The proof of Theorem~\ref{thm:decidingEqualityUnmarkedSpectrum} uses generating functions of the length spectrum. In our situation, these functions turn out to be (multivariate) rational functions. The appearance of such generating function when working with integral points in polyhedra is standard, see e.g.~\cite{barvinok}. Though the question of equality of these generating functions did not seem to have been considered from an algorithmic point of view before.

\begin{proof}
The idea is to encode the length spectrum by means of a generating series which turns out to be a rational function. So the question of equality of unmarked length spectra can be reduced to equality of generating series. However the question is less trivial as it seems as our rational functions admit a short representation (of size $O(n)$) but \emph{not} as a ratio of two polynomials.

Let us now explain how to encode the unmarked length spectrum of a weighted graph $(G,w)$ in a generating series. We let $\Lambda_{G,w}$ denote the length spectrum with multiplicity (as a submultiset of $\mathbb{R}$) and define
\[
f_{G,w}(s) := \sum_{\lambda \in \Lambda_{G,w}} e^{-s \lambda}.
\]
Such series is called a \emph{Dirichlet series} and is frequently encountered in number theory (e.g. the Riemann $\zeta$ function $\zeta(s) = \sum \exp(-\log(n) s)$) see e.g.~\cite{hardyriesz}. Because $\Lambda_{G,w}$ has at most quadratic growth (i.e. $\# \{\lambda \in \Lambda_{G,w} : \lambda < R\} \le c R^2$ for some $c > 0$), $f_{G,w}(s)$ defines an absolutely convergent series on the half plane $D = \{s \in \mC : \Re(s)  > 0\}$. By~\cite[Th. 13]{hardyriesz}, we have that $f_{G,w}(s) = f_{G',w'}(s)$ for all $s \in D$ if and only if $(G,w)$ and $(G',w')$ have the same length spectrum.

Recall that Algorithm~\ref{alg:algo1} produces in time $O(n^2 \log\log n)$ a list of tight simple cycles $\{c_i\}_{i=1,\ldots,t}$ such that $\langle c_i, c_{i+1} \rangle = 1$ and $\{[c_i] / w(c_i) \mid i=1,\ldots,t\}$ contains the extremal points of $\Ball{G,w}$. We refer the reader to Lemma~\ref{lem:unimodularPartition}
for the fist property.
We consider the half-open cones $C_i = \mZ_{\ge 0} [c_i] + \mZ_{> 0} [c_{i+1}] \subset H_1(T; \mZ)$. Then, using as in the proof of Theorem~\ref{thm:lengthSpectrumComputation} that $\Norm{G,w}$ is linear in $C_i$ and the fact that the cones are unimodular, the length spectrum of classes in $C_i$ has a rather simple generating series
\[
\sum_{\alpha \in C_i} e^{-s \Norm{G}(\alpha)}
=
\sum_{(x,y) \in \mZ_{\ge 0} \times \mZ_{> 0}}
e^{-s(\Norm{G}(x[c_i] + y[c_{i+1}])}
=
\frac{e^{-s\lambda_{i+1}}}{(1-e^{-s\lambda_i})(1-e^{-s\lambda_{i+1}})}.
\]
where $\lambda_i = w(c_i)$. Since the half-open cones form a partition of $H_1(T; \mZ) \setminus \{0\}$ we obtain that
\begin{equation} \label{eq:EhrhartShortForm}
f_{G,w}(s) =
\sum_{i=1}^t \frac{e^{-s\lambda_{i+1}}}{(1-e^{-s\lambda_i})(1-e^{-s\lambda_{i+1}})}
\end{equation}
Note that the number $t$ of summands in~\eqref{eq:EhrhartShortForm} is $O(n)$. We are now left with the question of comparing $f_{G,w}(s)$ and $f_{G',w'}(s)$. At this stage, it becomes very unconvenient to work with expressions involving exponentials. Instead we work with rational functions in $\mQ(z_1, \ldots, z_r)$. Recall that the weights of both $G$ and $G'$ are integral linear combinations of $o_1$, \ldots, $o_r$. It is hence also the case for $\lambda_i$ that are weights of simple cycles, i.e. $\lambda_i = \lambda_{i,1} o_1 + \ldots + \lambda_{i,r} o_r$ for some $\lambda_{i,j} \in \mZ$. For convenience we denote $z^{\lambda_i} = z_1^{\lambda_{i,1}} z_2^{\lambda_{i,2}} \cdots z_r^{\lambda_{i,r}}$. We define
\begin{equation} \label{eq:EhrhartShortFormz}
F_{G,w}(z_1, \ldots, z_r)
:=
\sum_{i=1}^t \frac{z^{\lambda_{i+1}}}{(1-z^{\lambda_i})(1-z^{\lambda_{i+1}})}
\end{equation}
so that $f_{G,w}(s) = F_{G,w}(e^{-so_1}, e^{-so_2}, \ldots, e^{-so_r})$. Since the functions $s \mapsto e^{-so_i}$ are algebraically independent (because the $o_i$ are linearly independent) we have that $f_{G,w}(s) = f_{G',w'}(s)$ for $s \in D$ if and only if $F_{G,w}(z_1, \ldots, z_r) = F_{G',w'}(z_1, \ldots, z_r)$ as rational functions.
Recall that $\lambda_{i,j} \in \mZ$ so that negative powers might appear in the expression~\eqref{eq:EhrhartShortFormz}. We can decompose $\lambda_i = \lambda_i^+ - \lambda_i^-$ so that $z^{\lambda_i} = z^{\lambda^+_i} z^{-\lambda^-_i}$ and both $z^{\lambda^+_i}$ and $z^{\lambda^-_i}$ have non-negative exponents.
Let us now reduce to a common denominator and get rid of the negative terms but keeping the factorized form~\eqref{eq:EhrhartShortFormz}
\begin{equation} \label{eq:EhrhartSameDenom}
F_{G,w}(z_1, \ldots, z_r) =
\frac{\displaystyle \sum_{i=1}^t z^{\lambda_{i+1}^+} z^{\lambda_i^-} \prod_{j \not\in \{i,i+1\}} (z^{\lambda^-_j} - z^{\lambda^+_j})}{\displaystyle \prod_{i=1}^t (z^{\lambda^-_i} - z^{\lambda^+_i})}
\end{equation}
Note that the total degree of the numerator and denominator are bounded by $\sum_{i=1}^t \sum_{j=1}^r |\lambda_{i,j}|$. Since each $c_i$ is a simple cycle in $G$, each of them involves at most once each weight so that $\sum_{i=1}^t \sum_{j=1}^r |\lambda_{i,j}| \le t \cdot \|w\|_o$.

We now consider the problem of comparison of $F_{G,w}$ and $F_{G',w'}$. The most important fact is that the number of monomials in $\mZ[z_1, \ldots, z_r]$ of total degree at most $d$ is equal to the binomial $\binom{d+r}{r}$. We use dense data-structures for our polynomials so that access and updates are $O(1)$.

When taking the cross-multiplication of the rational functions $F_{G,w}$ and $F_{G',w'}$ we obtain a sum of $t+t'$ terms of the form $z^a \prod_i (z^{b_i} - z^{c_i})$ where the number of factors in the product is at most $t+t'$ and each factor has total degree $\max(|b_i|, |c_i|) \le \max(\|w\|_o, \|w'\|_o)$. To expand this product, we compute $P_k = \prod_{i=1,\ldots,k} (z^{\lambda^-_i} - z^{\lambda^+_i})$ from $P_{k-1}$. The cost of one multiplication is $O(\binom{\max(\|w\|_o, \|w'\|_o)+r}{r})$. The total cost is $O\left( (t+t') \binom{\max(\|w\|_o, \|w'\|_o)+r}{r} \right)$. Since we have $t+t'=O(n)$ such terms, the total cost of the comparison via expansion is $O\left( n^2 \cdot \binom{\max(\|w\|_o, \|w'\|_o)+r}{r} \right)$. For fixed $r$, it is $O(n^2 \cdot \max(\|w\|_o, \|w'\|_o)^r)$ which proves the last claim of the theorem.

Now, there is a well-known probability approach for comparing the factorized form by evaluating at random points : the DeMillo-Lipton-Schwartz–Zippel lemma~\cite{demillolipton}, \cite{schwartz} and~\cite{zippel}. Let us recall that we want to test whether the difference $F - F'$ between the two generating series (in their factorized forms~\eqref{eq:EhrhartShortFormz}) is zero. As we already mentioned this corresponds to the equality to zero of some polynomial $P$ with $O(n)$ terms and total degree $d = O((t + t') \cdot (\|w\|_o + \|w\|'))$. We emphasize that the latter is exponential in the input size. The general idea is to evaluate the polynomial $P$ at random integer points. If we find an integer vector $(v_1, \ldots, v_r)$ such that $P(v_1, \ldots, v_r) \not= 0$ we have a witness that $P$ is non-zero.  The cost of the evaluation of $P(v_1, \ldots, v_r)$ is $O(n \times r \times \log(d))$ using fast exponentiation (and ignoring arithmetic cost of integer addition and multiplication). The DeMillo-Lipton-Schwartz–Zippel states that if $P$ is non-zero then the probability that a random integer vector with entries in $S^r$ where $S \subset \mathbb{Z}$ is such that $P(v) = 0$ is at most $d / |S|$. In other words, we would find a witness with probability at least $1 - \frac{d}{|S|}$. This probability is $> 1/2$ if we choose $S = \{-d+1, \ldots, d\}$ so that $|S| = 2d$. Taking an element at random in the set $S$ requires $O(\log d)$ bits which is linear in the input size. We need to take $r$ of them and this cost is quadratic in the input size. Overall the random generation and evaluation has a polynomial cost.
\end{proof}

 We provide in Figure~\ref{fig:isospectralGraphs} two toroidal graphs. Their unit balls can be seen in Figure~\ref{fig:markedVSUnmarked}:  they are not linearly equivalent but have the same unmarked length spectrum. This construction can be seen as a discrete analogue of the non-isometric but isospectral hyperbolic surfaces constructed in~\cite{vigneras}.

\begin{figure}[ht!]
    \centering
    \includegraphics{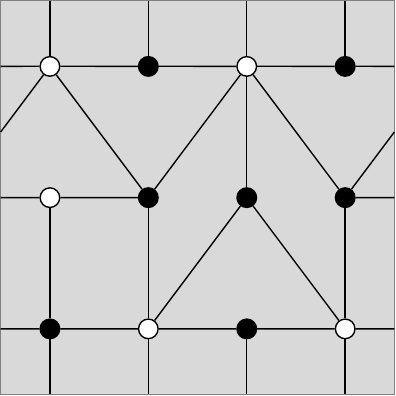} \hspace{1cm}
    \includegraphics{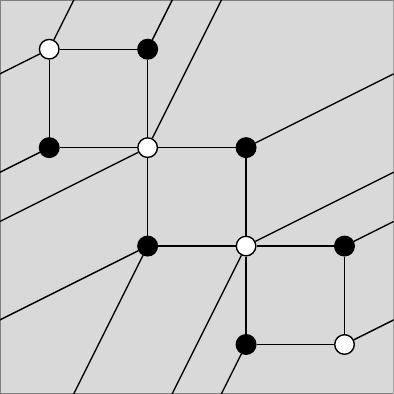}
    \caption{Two graphs $(G,w)$ (left picture) and $(G',w')$ (right picture) where all weights are $1/2$. Since the graphs are bipartite their norms are integral norms.}
    \label{fig:isospectralGraphs}
\end{figure}

\begin{figure}[ht!]
    \centering
    \includegraphics[width=.5\linewidth]{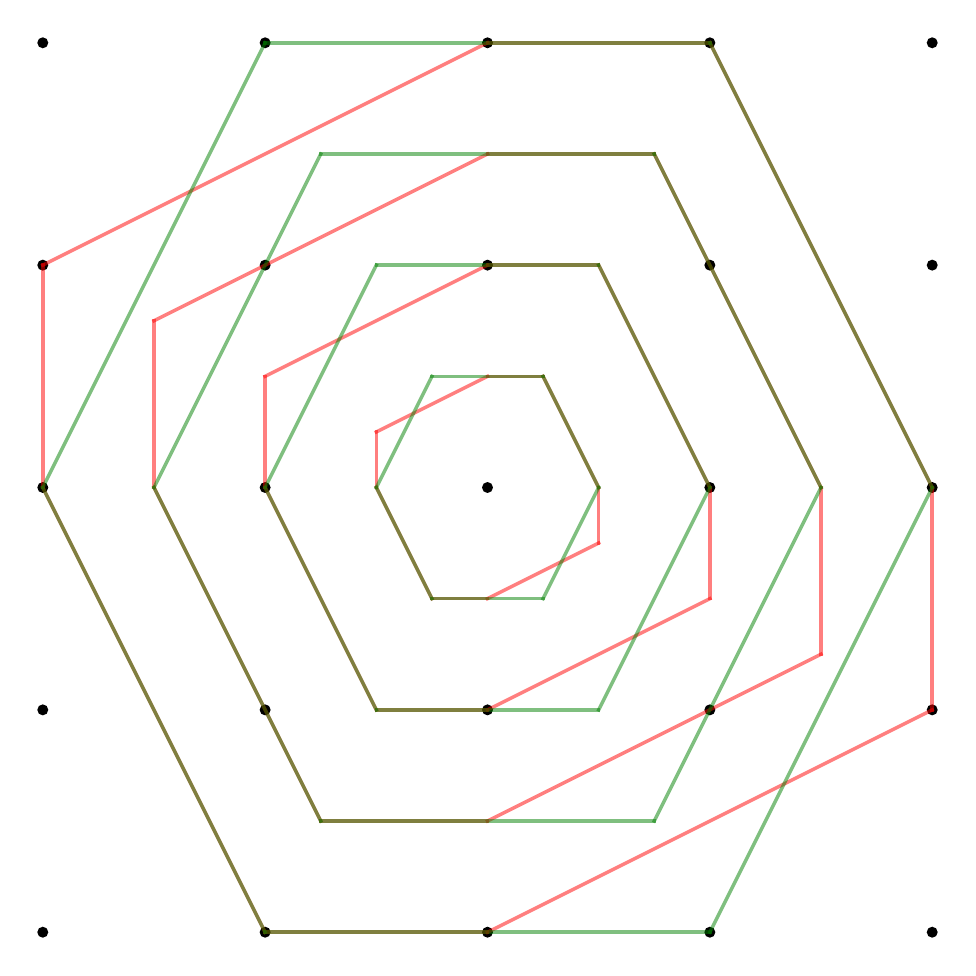}
    \caption{The four first dilates of the unit balls of the graphs $(G,w)$ and $(G',w')$ from Figure~\ref{fig:isospectralGraphs} respectively in green and orange. The black dots indicate the integer lattice. These two balls are not linearly equivalent: $\Ball{G,w}$ is a hexagon while $\Ball{G',w'}$ is an octagon. However, they have the same unmarked length spectrum, or equivalently the same Ehrhart series $\frac{z^6 + 3z^4 + 4z^3 + 3z^2 + 1}{(1 - z)^2\ (1 + z)^2\ (1 + z^2)} = 1 + 4z^2 + 4z^3 + 8z^4 + O(z^5)$.}
    \label{fig:markedVSUnmarked}
\end{figure}

\section*{Acknowledgements}
We would like to thank Marcos Cossarini who suggested to look at the norm point of view and Bruno Grenet who enlightened us about  the Polynomial Identity Testing problem. We would also like to thank the referees from SoCG 2023 for their helpful comments, in particular for suggesting the complexity improvement in the case of unweighted graphs.

\end{document}